\documentclass[11pt]{article}

\usepackage{amsmath,amsthm,amssymb,amsfonts,verbatim}
\usepackage{latexsym}
\usepackage{fullpage}

\usepackage{graphicx}
\usepackage{dcolumn}
\usepackage{bm}
\usepackage{amsmath,amsthm,amssymb,amsfonts,verbatim}
\usepackage{color}

\newtheorem{lemma}{Lemma}
\newtheorem{thm}[lemma]{Theorem}
\newtheorem{cor}[lemma]{Corollary}

\newtheorem{defn}{Definition}

\newtheorem{Remark}{Remark}
\newtheorem{exm}{Example}

\renewcommand{\epsilon}{\varepsilon}
\renewcommand{\le}{\leqslant}
\renewcommand{\ge}{\geqslant}

\def\F{\mathbb{F}}
\def\CC{\mathbb{C}}

\def\ZZ{\mathbb{Z}}
\def\PP{\mathbb{P}}
\def\MM{\mathbb{M}}

\def \mL {\mathcal{L}}

\def\cL{{\mathcal L}}

\def \cX {\mathcal{X}}

\def \Tr {{\rm Tr}}
\def\cA{\bar{A}}

\def \ba {{\bf a}}

\def \bc {{\bf c}}
\def \bx {{\bf x}}

\def \bh {{\bf h}}
\def \bu {{\bf u}}
\def \bv {{\bf v}}

\def \bo {{\bf 0}}
\def\supp {{\rm supp }}

\def \wt {{\rm wt}}

\def\bp{{\bar{\phi}}}
\def \res {{\rm res}}

\def\g{{\mathfrak{g}}}

\begin{document}

\date{}
\onecolumn
\title{Constructions of $k$-uniform states in heterogeneous systems}

\author{Keqin Feng,
\thanks{Keqin Feng is with Department of Mathematical Sciences, Tsinghua University, Beijing, China.}
Lingfei Jin
\thanks{Lingfei Jin is with School of Computer Science, Shanghai Key Laboratory of Intelligent Information Processing,  Fudan University, Shanghai 200433, China.}
Chaoping Xing and Chen Yuan
\thanks{Chaoping Xing and Chen Yuan are with the School of Electronic Information and Electrical Engineering, Shanghai Jiao Tong University, Shanghai 200240, China.}
\thanks{K. Feng is supported by the National Natural Science Foundation of China under Grants 12031011. }
\thanks{L. Jin is supported  by the National Natural Science Foundation of China under Grant  12271110,
National Key Research and Development Project 2021YFE0109900, and Shanghai Rising-Star Program (20QA1401100), Shanghai Science and Technology Program (Project No. 21JC1400600).}}

\maketitle


\begin{abstract}
A pure quantum state of $n$ parties associated with the Hilbert space $\CC^{d_1}\otimes \CC^{d_2}\otimes\cdots\otimes \CC^{d_n}$ is called $k$-uniform if all the reductions to $k$-parties are maximally mixed. The $n$  partite system is called homogenous if the local dimension $d_1=d_2=\cdots=d_n$, while it is called heterogeneous if the local dimension are not all equal. $k$-uniform sates play an important role in quantum information theory. There are many progress in characterizing and constructing $k$-uniform states in homogeneous systems. However, the study of entanglement for heterogeneous systems is much more challenging than that for the homogeneous case. There are very few results known for the $k$-uniform states in heterogeneous systems for $k>3$. We present two general methods to construct $k$-uniform states in the heterogeneous systems for general $k$. The first construction is derived from the error correcting codes by establishing a connection between irredundant mixed orthogonal arrays and error correcting codes. We can produce many new $k$-uniform states such that the local dimension of each subsystem can be a prime power. The second construction is derived from a matrix $H$ meeting the condition that $H_{A\times \bar{A}}+H^T_{\bar{A}\times A}$ has full rank for any row index set $A$ of size $k$.
These matrix construction can  provide more flexible choices for the local dimensions, i.e., the local dimensions can be any integer (not necessarily prime power) subject to some constraints. Our constructions imply that for any positive integer $k$, one can construct $k$-uniform states of a heterogeneous system in many different Hilbert spaces.
 \end{abstract}


\section{Introduction}

Quantum entanglement plays a fundamental important role in quantum information theory. It appears in many areas of quantum information theory including quantum communications~\cite{Be1993,Be1996}, quantum computing~\cite{Joz1997,Joz2002,Vir2005} and quantum key distribution~\cite{Koa2004}. Characterization of entanglement in multipartite quantum systems is an important open problem in quantum theory.
A pure multipartite quantum state of a system containing $n$ subsystems with local dimension $d$ is called $k$-uniform if all the reductions of the density matrix to $k$ subsystems are maximally mixed. Such states are also called maximally multipartite entangled states. From the perspective of information theory, a $k$-uniform state in the Hilbert space $(\CC^d)^{\otimes n}$ means that all information about the system is lost after removal of $n-k$ or more qudits.
Recently, great progress has been made on the characterization and construction of $k$-uniform states in homogeneous systems \cite{Feng,Goy14,Goy15,Pang19}. They can be constructed from orthogonal arrays \cite{Goy14}, Latin squares \cite{Goy15}, symmetric matrices \cite{Feng}, classical error correcting codes \cite{Feng} and quantum error correcting codes \cite{Scott}.

However, the quantum system is more likely to be heterogeneous in a practical setup where the local dimensions of the quantum system are mixed. For example, the physical systems for encoding may have different numbers of energy levels. Recently, several experiments were done on  physical systems consist of subsystems of different numbers of levels. It has been experimentally shown that a three partite state composed of one qubit and two qutrits
can be generated in laboratory, which exhibits genuinely multipartite entanglement \cite{Ma16}. The heterogeneous systems enable one to implement quantum steering more efficiently. This indicated that multipartite entangled states of heterogeneous systems can be implemented in the near future.

An $n$ partite system $\CC^{d_1}\otimes \CC^{d_2}\otimes\cdots\otimes \CC^{d_n}$ is called heterogeneous if the local dimensions $d_i$ are not all equal for $i=1,2,\cdots,n$. We can consider homogenous systems as a special case of the  heterogeneous systems $\CC^{d_1}\otimes \CC^{d_2}\otimes\cdots\otimes \CC^{d_n}$ when $d_1=d_2=\cdots=d_n$. Similarly, the study of $k$-uniform states in  heterogeneous systems is of great interest. As mentioned in \cite{Pang19}, the higher the uniformity of the multipartite entangled states, the more advantages they can offer. Therefore, constructing $k$-uniform states with higher uniformity in heterogeneous systems becomes an interesting and hot topic recently \cite{Hu13, Sun15,Chen06,Yu08, Goy16, Shi, Pang}. However, the theory of quantum entanglement in heterogeneous systems is far from satisfactory. The study of entanglement for heterogeneous systems is much more challenging than that for the homogeneous case. The characterization of $k$-uniform sates in heterogeneous systems is quite hard and it lacks efficient mathematical tools. There are only few results known on $k$-uniform sates in heterogeneous systems.

\subsection{Known results}

In general, a $k$-uniform state do not exist if the product of the dimensions of the $k$ local
Hilbert spaces is larger than that of the complementary system. Thus, we can say that a necessary condition for the existence of a $k$-uniform state in a heterogeneous system $\CC^{d_1}\otimes \CC^{d_2}\otimes\cdots\otimes \CC^{d_n}$ is $d_1d_2\cdots d_k\le d_{k+1}\cdots d_n$. In a particular case when $d_i$ are all equal, this necessary condition reduces to $k\le \lfloor\frac{n}{2}\rfloor$ where $\lfloor\cdot\rfloor$ is the floor function.

There are some particular cases of quantum entanglement in heterogeneous systems that have been studied, such as in tripartite systems $\CC^2\otimes\CC^2\otimes\CC^t$ \cite{Yu08,Mi04} and $\CC^2\otimes\CC^{t_1}\otimes\CC^{t_2}$ \cite{Wang13} and four-partitle systems $(\CC^2)^{\otimes 3}\otimes\CC^t$ \cite{Chen06}. The characterization of entanglement in heterogeneous systems appears much more hard for the lack of suitable mathematical tools.
Recently, some constructions of $k$-uniform states in heterogeneous systems are given especially for small $k$, for example $k=1,2,3$.
The relationship between orthogonal arrays and $k$-uniform sates in homogeneous systems has been well-studied \cite{Goy14,Shi20}. Mixed  orthogonal arrays (MOAs, or called asymmetric orthogonal arrays) are a natural generalization of orthogonal arrays (OAs) and can be applied to the construction of $k$-uniform sates in heterogeneous systems. In \cite{Goy16}, the authors established a link between mixed orthogonal arrays and
multipartite quantum states in heterogeneous systems and demonstrated the existence
of a wide range of highly entangled states in heterogeneous systems. In particular, they presented explicit constructions of $1$-uniform and $2$-uniform states in heterogeneous systems for any number of parties by introducing irredundant mixed orthogonal arrays (IrMOAs). A construction of irredundant mixed orthogonal arrays  from different schemes was given.
Later, a sufficient and necessary condition for the existence of $1$-uniform states was given in \cite{Br18}.
Recently, Shi et.al. studied the construction of $2$-uniform and $3$-uniform states from IrMOAs \cite{Shi}.
Besides, two methods of generating $(k-1)$-uniform states from $k$-uniform states were given.
In \cite{Pang}, the authors generalized difference scheme method and orthogonal partition method for constructing IrOAs  to IrMOAs. As a result, explicit constructions of $2$ and $3$-uniform states for arbitrary heterogeneous multipartite systems with coprime individual levels were given using IrMOAs.

Though some interesting results have been given on $k$-uniform states in heterogeneous systems, these results are far from satisfactory. Many questions related to quantum entangled states in heterogeneous systems remain open. For example, it was pointed out in \cite{Goy16, Pang} that it is not known for which heterogeneous systems $k$-uniform states exists, especially for five partite systems and $k\ge 3$. In literature, few is known on $k$-uniform states for $k>3$.
 We list some of the known results on the existence of $k$-uniform states for small $k$ in Table 1 to our knowledge.
We can also observe that most of the known $k$-uniform states in heterogeneous systems are derived from IrMOAs. Therefore, it is desired to develop some new universal approaches for producing $k$-uniform states in heterogeneous systems.

\begin{table}[]
	\caption{\bf Some known results on existence of $k$-uniform states for small $k$}
\bigskip
	\centering
\begin{tabular}{||p{3cm}|p{8cm}|p{2cm}||} \hline\hline
 $(\CC^d)^{\otimes t_1}\otimes (\CC^2)^{\otimes t_2}$ &Existence  & Reference\\
\hline
2-uniform & $d=3, t_2=1,t_1=7$ & \cite{Goy16}  \\
2-uniform & $d=3, t_2=1,t_1=13$  & \cite{Goy16}  \\
2-uniform & $d=5, t_2=1,t_1=11$  & \cite{Goy16}  \\
2-uniform & $d\ge 2, t_2=1,t_1\ge 5$ & \cite{Shi}  \\
      2-uniform &  $d\ge 2, t_2=2,t_1\ge 7$, $t_1\neq 4d+2,4d+3$   &  \cite{Shi} \\
     2-uniform & $d=3,3\le t_2\le 4$,$t_1\ge 7$   & \cite{Shi}  \\
       2-uniform & $d=3,5\le t_2\le 11$,$t_1\ge 6$    & \cite{Shi}  \\
       2-uniform & $d=5,3\le t_2\le 4$,$t_1\ge 7$,$t_1\neq 22,23$   & \cite{Shi}  \\
       2-uniform & $d=5,5\le t_2\le 8$,$t_1\ge 6$,$n\neq 22$    & \cite{Shi}  \\
       2-uniform & $d=5,9\le t_2\le 19$,$t_1\ge 6$,$t_1\neq 21,22$    & \cite{Shi}  \\
       2-uniform & $d=7,3\le t_2\le 5$,$t_1\ge 7$,$t_1\neq 30,31$   & \cite{Shi}  \\
       2-uniform & $d=7,6\le t_2\le 12$,$t_1\ge 6$,$t_1\neq 30$    &  \cite{Shi} \\
       2-uniform & $d=7,13\le t_2\le 27$,$t_1\ge 6$,$t_1\neq 29,30$    &  \cite{Shi} \\
       \hline
            3-uniform & $d>4$ odd prime power, $4\le t_1\le d$,$2d^2\le t_2\le 4d^2$,$t_2\ge 4d^2$ & \cite{Pang} \\
      3-uniform & $d=5,4+5\times 100^m\le t_2\le 6\times 100^m$,$0\le t_1\le 6^m$ & \cite{Shi} \\
  \hline
        $(\CC^3)^{\otimes t_1}\otimes (\CC^2)^{\otimes t_2}$ &Existence  & Reference \\
        \hline
        2-uniform & $t_1=1$,$t_2\ge 8$   &  \cite{Pang} \\
        2-uniform & $t_1=2$,$t_2\ge 12$   &  \cite{Pang} \\
        2-uniform & $t_1=3$,$t_2\ge 11$   &  \cite{Pang} \\
          2-uniform & $t_1=4$,$t_2\ge 10$   &  \cite{Pang} \\
          2-uniform & $1\le t_1\le 2^m 3$,$t_2\ge 2^{m-1}3^2+3$,$m\ge 2$   &  \cite{Pang} \\ \hline
           3-uniform & $t_1=5$,$t_2\ge 16$  &  \cite{Pang}  \\
           3-uniform & $t_1=4$,$t_2\ge 16$   &  \cite{Pang} \\
   3-uniform & $4+7\times 36^m\le t_2\le 9\times 36^m$,$0\le t_1\le 4^m$  & \cite{Shi}  \\
   \hline
   \hline
\end{tabular}
\end{table}
Among those $k$-uniform states, there is a special class of  $k$-uniform states which is called absolutely maximally entangled (AME) states that are of great interest.
Suppose $|\Phi\rangle$ is a $k$-uniform sate in $(\CC^d)^{\otimes n}$, then we have  $k\le \lfloor\frac{n}{2}\rfloor$ by the Schmidt decomposition. We call a $k$-uniform state an absolutely maximally entangled (AME) state when the equality is achieved, i.e., $k=\lfloor\frac{n}{2}\rfloor$.  AME states are important resources which can be used in threshold quantum secret sharing schemes \cite{H12}, parallel teleportation \cite{H13} and holographic quantum codes \cite{Pa15}. However, the existence of AME states are very rare for given local dimensions. In \cite{Scott}, it was shown that there is a one to one correspondence between the AME state in the homogeneous system  and quantum error correcting code. However, not all homogeneous systems contain AME states. For example, AME states of seven qubits do not exist \cite{Hu17}.

It is known that for a general system $\CC^{d_1}\times \CC^{d_2}$ with $d_1\neq d_2$, AME state do not exist. Therefore, AME states do not exist for heterogeneous systems consisting of $n$ subsystems where $n$ is even \cite{Scott}. In fact, for an AME state in a heterogeneous system $\CC^{d_1}\otimes \CC^{d_2}\otimes\cdots\otimes \CC^{d_n}$ where $d_i$ are not all equal, then $d_1\cdots d_{\lfloor\frac{n}{2}\rfloor}\le d_{\lfloor\frac{n}{2}\rfloor+1}\cdots d_n$ which implies that $n$ must be odd.
 In \cite{Shen20}, AME states in tripartite heterogeneous systems were studied. More precisely, the existence of AME states in tripartite heterogeneous systems were investigated and some AME states were constructed in such tripartite heterogeneous systems. They established the connection between heterogeneous AME states and multi-isometry matrices. Some nonexistence of AME states were given by computer search \cite{Shi}. For example, it was proved that AME states do not exist in $\CC^3\otimes (\CC^2)^{\otimes8}$ \cite{Shi}.  Generally speaking, few is known on the existence and nonexistence of AME states in heterogeneous systems.



\subsection{Our main results and techniques}
A key problem in characterizing entanglement in heterogeneous systems is the lack of suitable
mathematical tools. For example, Galois fields do not exist in non-prime dimensions. This makes the
study of entanglement in heterogeneous systems more challenging than that in homogeneous systems.
It is also important to note that, it is still not possible to determine in
which heterogeneous systems $k$-uniform states exist.

In this paper, we present two general constructions of $k$-uniform states in heterogeneous systems. The first construction makes use of the connection between IrMOAs and $k$-uniform states. By giving a general construction of IrMOAs from error correcting codes, we are able to produce families of $k$-uniform states in heterogeneous systems where each level is a prime power.
To instantiate this construction, we consider  generalized algebraic geometry codes whose components are the evaluations of a function at  places of degree $m_i$. For a function $f\in \cL(G)$, an algebraic geometry code usually encode a message to $(f(P_1),\ldots,f(P_n))$ where $P_1,\cdots,P_n$  are $n$ distinct rational points. In our construction, we choose $n$
 places $P_1,\cdots,P_n$ with degree $m_i$ and evaluate each function on $P_1,\cdots,P_n$ such that $f(P_i)\in \F_{q^{m_i}}$. Thus we obtain a generalized algebraic geometry code over mixed alphabets. Furthermore, we give  a lower bound on the minimum distance of the generalized algebraic geometry code and its dual code.
The second construction is totally different from the existing tools for producing $k$-uniform states in heterogeneous systems. Inspired by the symmetric matrix construction of $k$-uniform states in homogeneous systems in \cite{Feng}, we generalized the idea of using symmetric matrix for the construction of $k$-uniform states in heterogeneous systems. We consider an $n$-qudit state $|\Phi\rangle=\frac{1}{\sqrt{D}}\sum_{\bc\in G }\phi(\bc)|\bc\rangle$ in $\bigotimes_{i=1}^n\CC^{d_i}$  with $\phi(\bc)=\zeta_D^{\bc H\bc^T}$ where $D=$lcm$(d_1,\ldots,d_n)$. We manage to show that if the submatrix $H_{A\times \bar{A}}+H_{\bar{A}\times A}^T$ of $H$ has full row rank for each $k$-subset $A\subseteq \{1,\ldots,n\}$ and $\bar{A}=\{1,\ldots,n\}\setminus A$, then this $n$-qudit state is $k$-uniform. If $H$ is symmetric, this condition can be simplified.
We present two constructions based on this method, i.e., Construction I and Construction II.
Construction I is obtained from symmetric matrices. More explicitly, we give examples of $1$-uniform states and AME states. Construction II is obtained from random matrices which proceeds in two steps.
We first assume $d_i=p^{\alpha_i}$ and map $\ZZ_{d_i}$ to $\ZZ_p^{\alpha_i}$. The merit of this map is that the matrix under this map is defined over prime field $\ZZ_p$. This greatly simplifies the rank calculation of a matrix. Then, this construction is extended to any number $d_i=\prod_{j=1}^t p_j^{\alpha_{ij}}$. To do this, we decompose this case into $t$ prime cases $\ZZ_{p_j}^{\alpha_{ij}}$. By Chinese Reminder Theorem, we merge them together and obtain a matrix over $\ZZ_d$ with $d=p_1p_2\cdots p_t$. Construction II provides a rather flexible parameter ranges for the $k$-uniformity. The details of the construction can be found in Section 4.


\subsection{Organization of this paper}
This paper is organized as follows. In Section 2, we briefly introduce some definitions and notations on $k$-uniform states and function fields. In Section 3, we define codes over mixed alphabets and their dual codes. The relationship between codes over mixed alphabets and  $k$-uniform stats in heterogeneous systems is given. Then, we present an explicit construction of  $k$-uniform stats in heterogeneous systems from generalized algebraic geometry codes. Finally, we give a universal construction of $k$-uniform stats in heterogeneous systems by using matrices in Section 4. We present two constructions in Section 4. Construction I is obtained from symmetric matrices and Construction II is obtained from random matrices. We conclude our results in Section 5.

\section{Preliminaries}
In this section we introduce some notations and definitions on $k$-uniform states, function fields and algebraic geometry codes.

\subsection{$k$-uniform states in heterogeneous systems}
We denote a $d$-dimensional space by $\CC^d$. In a general setting, an  $n$-partite quantum system is associated with the Hilbert space $\CC^{d_1}\otimes \CC^{d_2}\otimes\cdots\otimes \CC^{d_n}$  where $d_{1},\ldots,d_{n}\ge 2$ are integers. If the local dimensions $d_i$ are equal for $i=1,2,\cdots,n$, then we call it a  homogeneous system, denoted as $(\CC^d)^{\otimes n}=\CC^d\otimes \cdots\otimes \CC^d$. If the local dimensions are mixed, i.e., $d_i$ are not all equal for $i=1,2,\cdots,n$, then we call it a heterogeneous  system $\CC^{d_1}\otimes \CC^{d_2}\otimes\cdots\otimes \CC^{d_n}$.  We assume  $d_1\ge d_2\ge \cdot\cdot\cdot\ge d_n\ge 2$ in such a heterogeneous system.

Denote $[n]$ as the set of integers $\{1,2,\cdots,n\}$. Denote $\binom{[n]}{k}$ as the collections of all $k$-subsets in $[n]$.
Denote by $\gcd(d_1,\ldots,d_n)$ the greatest common divisors of $d_1,\ldots,d_n$ and by lcm$(d_1,\ldots,d_n)$ the least common multiple of $d_1,\ldots,d_n$.
Let $\{|\ba\rangle:\ba\in \ZZ_{d_i}=\mathbb{Z}/d_i\ZZ=\{0,1,\ldots,d_{i}-1\}\}$ be a standard orthogonal basis. Denote $G= \ZZ_{d_1}\oplus \ldots \oplus \ZZ_{d_n}=\bigoplus_{i\in [n]}\ZZ_{d_i}$.
For any subset $A\subset \left[ n\right] =\left\{ 1,2,\ldots ,n\right\} $, let $G_A=\bigoplus_{i\in A} \ZZ_{d_i}$.
A pure quantum state in $\CC^{d_1}\otimes \CC^{d_2}\otimes\cdots\otimes \CC^{d_n}$  can be represented as $|\Phi\rangle=\sum_{\bc\in G}\phi(\bc)|\bc\rangle$.
The density matrix of a quantum state  $|\Phi\rangle=\sum_{\bc\in G}\phi_{\bc}|\bc\rangle$ is defined by $\rho:=|\Phi\rangle\langle\Phi|=\sum_{\bc,\bc'\in G}\phi_{\bc}\bp_{\bc'}|\bc\rangle\langle\bc'|$. For a subset $A$ of $[n]$ and a vector $\bc\in G$, we denote by $\bc_A$ the projection of $\bc$ at $A$.  The reduction of $|\Phi\rangle$ to $A$ has the density matrix $\rho_A:=\sum_{\bc,\bc'\in G}\phi_{\bc}\bp_{\bc'}\langle\bc_{\cA}|\bc'_{\cA}\rangle|\bc_A\rangle\langle\bc'_A|$, where $\cA$ is the complement set of $A$ (i.e., $\cA=\{1,2,\dots,n\}\setminus A$) and $\langle\bc_{\cA}|\bc'_{\cA}\rangle$ is defined to be $1$ if $\bc_{\cA}=\bc'_{\cA}$ and $0$ otherwise.


Thus the reduction of the density matrix
$\rho$ to its subsystem indexed by $A$ is
\begin{eqnarray}
\rho_A&:=&\sum\limits_{\bc,\bc'\in G}\phi(\bc)\bar{\phi}(\bc')\langle\bc_{\bar{A}} | \bc'_{\bar{A}}\rangle|\bc_A
\rangle\langle \bc'_A|\\
&=&\sum\limits_{\bc_A,\bc'_A\in G_A}\left(\sum\limits_{\bc_{\bar{A}},\bc'_{\bar{A}}\in G_{\bar{A}}}\phi(\bc_A,\bc_{\bar{A}})\bar{\phi}(\bc'_A,\bc'_{\bar{A}})\langle \bc_{\bar{A}} | \bc'_{\bar{A}}\rangle|\bc_A\rangle\langle \bc'_A|\right)\\
&=&\sum\limits_{\bc_A,\bc'_A\in G_A}\left(\sum\limits_{\bc_{\bar{A}}\in G_{\bar{A}}}\phi(\bc_A,\bc_{\bar{A}})\bar{\phi}(\bc'_A,\bc_{\bar{A}})|\bc_A\rangle\langle \bc'_A|\right)
\end{eqnarray}

Now we are ready to give a definition of $k$-uniform states.
\begin{defn}\label{defn:1}
For $k\geq1$, a pure quantum state  $|\Phi\rangle=\sum_{\bc\in G}\phi_{\bc}|\bc\rangle$ is called $k$-uniform, if for any $ A\subseteq[n]$ with $|A|=k$, $\rho_A=\frac{1}{d_A}\sum\limits_{\bc_A\in G_A}|\bc_A\rangle\langle \bc_A|$, and $d_A=\prod\limits_{i\in A}d_i=|G_A|$, i.e.,

\[\sum\limits_{\bc_{\bar{A}}\in G_{\bar{A}}}\phi(\bc_A,\bc_{\bar{A}})\bar{\phi}(\bc'_A,\bc_{\bar{A}})=
\begin{cases}
0& \mbox{if $ \bc_A\neq \bc'_A$} \\
1/d_A& \mbox{if $\bc_A=\bc'_A$}
\end{cases}
(\forall \bc_A,\bc'_A\in G_A)
\]

\end{defn}

It is well known that $k\le\lfloor\frac{n}{2}\rfloor$ due to the Schmidt decomposition. If  $k=\lfloor\frac{n}{2}\rfloor$, then the quantum state  $|\Phi\rangle$ is called an absolutely maximally entangled (AME) state.

Orthogonal arrays have wide applications and  have been used a lot to construct $k$-uniform states in homogenous systems. Similarly, as a generalization of orthogonal arrays, mixed orthogonal arrays have been applied to the construction of $k$-uniform states in heterogeneous systems and have turned out to be a powerful and convenient tool.
\begin{defn}
A mixed orthogonal array (MOA)  $MOA(r,(d_1,d_2,\ldots,d_n),k)$ is an array of $r$ rows and $n$ columns satisfying the following conditions
\begin{enumerate}
  \item The $j$-th column has symbols from the set $S_j=\{0,1,2,\cdots,d_j-1\}$.
  \item For any subset $A=(j_1,\ldots,j_k)\subseteq(1,\ldots,n)$, $j_1<\cdots<j_k$ with size $k$, every possible $k$-tuple occurs the same number of times as a row in the subarray where the columns are indexed by $A$, i.e., each element in $S_{j_1}\times\ldots\times S_{j_k}$ appears $\frac{r}{d_{j_1}d_{j_2}\cdots d_{j_k}}$ times.
\end{enumerate}
 Furthermore, a mixed orthogonal array is called irredundant if every subset of $n-k$ columns contains a different sequence of $n-k$ symbols in every row, denoted as $IrMOA(r,(d_1,d_2,\ldots,d_n),k)$.
\end{defn}
It is obvious to see that if $d_1=d_2=\cdots=d_n$, then a mixed orthogonal array is reduced to an orthogonal array. However, it is not true that any MOA can be used to construct a $k$-uniform state in heterogenous system. The relationship between irredundant mixed orthogonal array and $k$-uniform states is characterized  in \cite{Goy16}. It is worthy to note that in general it is not easy to check whether an MOA is an IrMOA.

\begin{thm}\cite{Goy16}\label{thm:2.1}
Let $(m_{ij})_{1\leq i\leq r,1\leq j\leq n}$ be an $IrMOA(r,(d_1,\ldots ,d_n),k)$, then
$|\Phi\rangle=\frac{1}{\sqrt{r}}\sum\limits_{i=1}^{r}|m_{i1}\cdots m_{in}\rangle$ is a $k$-uniform state in $\mathbb{C}^{d_1}\otimes \ldots \otimes \mathbb{C}^{d_n}$.
\end{thm}

Therefore, as long as one can construct an irredundant mixed orthogonal array, then one can produce a $k$-uniform state. For example, the following matrix

\begin{equation*}
\left(
                                                                               \begin{array}{ccccc}
                                                                                 0 & 0&0& 0&0 \\
                                                                                 0 & 1&2& 1&1 \\
                                                                                 1& 1&1& 2&0 \\
                                                                                1 & 2&0& 0&1 \\
                                                                               2& 2&2& 1&0 \\
                                                                                2 & 0&1& 2&1 \\

                                                                               \end{array}
                                                                             \right).
\end{equation*}
is an $IrMOA(6,(3,3,3,3,2),1)$. This is a genuinely mixed orthogonal array.
Then we can construct a $1$-uniform state $|\Phi\rangle=\frac{1}{\sqrt{6}}(|00000\rangle+|01211\rangle+|11120\rangle+|12001\rangle+|22210\rangle+|20121\rangle)$
in $(\CC^3)^{\otimes 4}\otimes\CC^2$.

\subsection{Backgrounds on function fields}

As a preparation, we recall some basic results on algebraic geometry (AG) codes first. Readers may refer to \cite{St} for more details.

Let $\cX$ be a smooth, projective, absolutely irreducible curve of genus $\g(\cX)$ (we will use $\g$ instead of $\g(\cX)$ if there is no confusion in the context) defined over $\F_q$. We denote by $F:=\F_q(\cX)$ the function field of $\cX$ over $\F_q$. An element of $F$ is called a function. A place $P$ of $F$ is the maximal ideal in a valuation ring $O$ of $F/\F_q$. The residue field $O/P$ is isomorphic to  an extension field over $\F_{q}$. The degree of  $P$ is defined to be $[O/P:\F_q]$. We denote $\PP_F$ as the set of all places of $F$.
Let $\{P_1,...,P_n\}$ be a set of  $n$ distinct  places over $\F_q$ where $\deg(P_i)=m_i$  for $i=1,\ldots,n$. A place of degree one is called rational. A divisor $G$ of $F$ is of the form $\sum_{P\in \PP_F}m_pP$ with finitely many nonzero coefficients $m_P$. The degree of $G$ is defined by  $\deg(G)=\sum_{P\in\PP_F}m_p\deg(P)$. The support of $G$, denoted by  $\supp(G)$, is defined to be $\{P\in\PP_F: \; m_P\neq 0\}$. Thus $\supp(G)$ is a finite set. Let $G$ be a divisor of degree $m$  such that ${\rm Supp}(G)\cap\{P_1,\cdots,P_n\}=\emptyset$.

For a divisor $G$, the Riemann-Roch space associated to $G$ is
\[\mL(G)=\{f\in F\setminus\{0\}:\; (f)+G\geq0\}\cup\{0\}.\]
Then $\mL(G)$ is a finite-dimensional vector space over $\F_q$ and we denote its dimension by $\ell(G)$. By the Riemann-Roch theorem we have
\[\ell(G)\geq \deg(G)+1-g,\]
where the equality holds if $\deg(G)\geq 2g-1$.

Now we are ready to give the definition of algebraic geometry (AG) codes. Let $P_1,\dots,P_n$ be pairwise distinct rational places of $F$ and let $D=P_1+\dots+P_n$. Choose a
divisor $G$ of $F$ such that ${\rm supp}( G)\bigcap {\rm supp}( D)=\emptyset$. Then $\nu_{P_i}(f)\geq0$ for all $1\leq i\leq n$ and for any $f\in \mL(G)$. Consider the map

\[\Psi:\mL(G)\rightarrow \F_q^n,\quad f\mapsto(f(P_1),\dots,f(P_n)).\]
Obviously the image of  $\Psi$ is a $q$-ary linear code, denoted by $C_\cL(D,G)$.  $C_\cL(D,G)$ is usually called the functional algebraic-geometry code. If $\deg(G)<n$, then
$\Psi$ is an embedding and we have $\dim(C_\cL(D,G))=\ell(G)$.
By the Riemann-Roch theorem we can estimate the parameters of an AG code
(see \cite[Theorem 2.2.2]{St}).

\begin{lemma}\label{lem:2.1} $C_\cL(D,G)$ is an $[n,k,d]$-linear code over $\F_q$ with parameters
\[k=\ell(G)-\ell(G-D),\quad  d\geq n-\deg(G).\]
\begin{itemize}
\item[{\rm (a)}] If $G$ satisfies $\g\leq \deg(G)<n$, then
\[k=\ell(G)\geq \deg(G)-\g+1.\]
\item[{\rm (b)}] If additionally $2\g-2<\deg(G)<n$, then $k=\deg(G)-\g+1$.
\end{itemize}
\end{lemma}

The differential space of $F$ is defined to be
\[\Omega_F:=\{fdx:f\in F\},\]
This is a one-dimensional space over $F$. For a  place $P$ and a function $t$ with $\nu_P(t)=1$, we define $\nu_P(fdt)=\nu_P(f)$. The divisor associated with a nonzero differential $\eta$ is defined to be ${\rm div}(\eta)=\sum_{P\in\PP_F}\nu_P(\eta)P$. Such a divisor is called a canonical divisor. Every  canonical divisor has degree $2\g-2$. Furthermore, any two canonical divisors are equivalent.  Now, if $P$ is a rational place and $\nu_P(fdt)\ge -1$,  we define the residue of $fdt$ at $P$ to be $(ft)(P)$, denoted by $\res_{P}(fdt)$.

 For a divisor $G$, we define the $\F_q$-vector space
\[\Omega(G)=\{\eta\in \Omega_F\setminus\{0\}:\;{\rm div}(\eta)\ge G\}\cup\{0\}\]
and denote the dimension of $\Omega(G)$ by $i(G)$. Then one has the following relationship
\[i(G)=\ell(K-G),\]
where $K$ is a canonical divisor.

 We define another algebraic geometry code $C_\Omega(D,G)$ as
\[C_\Omega(D,G)=\{(\res_{P_1}(\eta),\dots,\res_{P_n}(\eta)):\;\eta\in
\Omega(G-D)\}.\]
It is well known that $C_\Omega(D,G)$ is the Euclidean dual of  $C_\cL(D,G)$.
We have the following results \cite[Theorem 2.2.7 and Proposition 2.2.10]{St}.
\begin{lemma}\label{lem:2.2}  $C_\Omega(D,G)$ is an $[n,k^{\perp},d^{\perp}]$-linear code over $\F_q$ with parameters
\[k^{\perp}=i(G-D)-i(G),\quad  d^{\perp}\geq \deg( G)-(2\g-2).\]
\begin{itemize}
\item[{\rm (a)}] If $G$ satisfies $2\g-2<\deg(G)<n$, then
\[k^{\perp}= n-\deg(G)+\g-1.\]
\end{itemize}
\end{lemma}
To distinguish two classes of algebraic geometry codes $C_\cL(D,G)$ and $C_\Omega(D,G)$, we usually call $C_\cL(D,G)$ a functional AG code and $C_\Omega(D,G)$ a differential or residual AG code.

\section{Construction from error correcting codes}
It is known that there is  a close relationship between $k$-uniform states in homogeneous systems and linear codes \cite{Feng}.
In this section, we generalize the idea in \cite{Feng} to construct $k$-uniform states in heterogeneous systems by defining codes over mixed alphabets. In fact, codes over mixed alphabets have been studied in \cite{Feng06}. In this section, we define the dual code of the code and show some of the probabilities. Furthermore, we define the generalized algebraic geometry codes and give lower bounds on the distance. It turns out that our codes can produce IrMOAs and therefore can give $k$-uniform states.

\subsection{A general construction from codes }
We first present a general construction of $k$-uniform states in heterogenous systems from error correcting codes. We first define a code $C$ over mixed alphabets, and its dual code $C^{\perp}$. By investigating the properties of the generator matrix of $C$, we can see that one can construct an IrMOA from the generator matrix of $C$ with some restrictions on its distance and the dual distance. Using the relationship between IrMOA and $k$-uniform states in Theorem \ref{thm:2.1}, one can obtain $k$-uniform states in heterogenous systems.

Let $q$ be a power of a prime number $p$ and $\F_q$ be the finite field with $q$ elements.
Let $m_1,m_2,\ldots,m_n$ be $n$ integers with $m_i\geq1$ for $1\le i\le n$.
Let $A$ be a subset of $[n]$, for a vector $\bv=(v_1,\cdots,v_n)$, denote the projection of the vector $\bv$ at $A$ as $\bv_A=(v_i)_{i\in A}$.  Assume that $\F_{q^{m_i}}/\F_q$ is a field extension with $[\F_{q^{m_i}}:\F_q]=m_i$.

We denote by $\Tr_i$ the trace function from $\mathbb{F}_{q^{m_i}}$ to $\mathbb{F}_{q}$, i.e.,
\[\Tr_i:\alpha\rightarrow\alpha+\alpha^q+\cdots+\alpha^{q^{m_i-1}}.\]
For any given $\F_q$-basis $\{\alpha_{i_1},\ldots,\alpha_{im_i}\}$ of $\mathbb{F}_{q^{m_i}}$, it is well known that there exists a dual $\F_q$-basis $\{\beta_{i_1},\ldots,\beta_{im_i}\}$ (see \cite{Nie}), i.e.,
$$
\Tr(\alpha_{ij}\beta_{it})=\left\{
            \begin{array}{ll}
              1,& \hbox{for\ $j=t$,}\\
              0,& \hbox{for\ $j\neq t$.}
            \end{array}, \quad 1\le t,j\le m_i.
          \right.
$$

Consider the  $\mathbb{F}_{q}$-vector space $\bigoplus_{i=1}^n\mathbb{F}_{q^{m_i}}$.
For a vector $\bu=(u_1,\ldots,u_n)$ in $\bigoplus_{i=1}^n\mathbb{F}_{q^{m_i}}$, the Hamming weight of $\bu=(u_1,\ldots,u_n)$  is defined to be the number of nonzero positions, denoted by $\wt(\bu)=\#\{i| 1\le i\le n, u_i\neq 0\in\F_{q^{m_i}}\}$. For two vectors $\bu,\bv\in\bigoplus_{i=1}^n\mathbb{F}_{q^{m_i}}$, define the Hamming distance of $\bu,\bv$ as  $d(\bu,\bv)=\wt(\bu-\bv)$.
An $\mathbb{F}_{q}$-vector space $C\subseteq\bigoplus_{i=1}^n\mathbb{F}_{q^{m_i}}$ is called an $\mathbb{F}_{q}$-linear code of length $n$, dimension $k=\dim_{\F_q} C$.
The minimum distance of the code $C$ is defined as the minimum Hamming distance, i.e.,
\[d=\min \{d(\bu,\bv)|\bu,\bv\in C,\quad \bu\neq\bv\}=\min\{ \wt(\bc)|\bc\neq0\in C\}\]

We define the dual code $C^\bot$ of $C$ as
$$C^\bot=\{\bu=(u_1,\ldots,u_n)\in\bigoplus_{i=1}^n\mathbb{F}_{q^{m_i}}:\sum_{i=1}^n\Tr_i(u_ic_i)=0,\quad\forall \bc=(c_1,\ldots,c_n)\in C\}.$$
We have the following relationship between $C$ and $C^{\perp}$.
\begin{lemma}\label{lem:3.1}
$\dim_{\F_q} C+\dim_{\F_q} C^\bot=\sum_{i=1}^nm_i$.
\end{lemma}
\begin{proof}
Let $\alpha_{i_1},\ldots,\alpha_{im_i}$ be a basis of $\mathbb{F}_{q^{m_i}}/\mathbb{F}_{q}$. Let $\beta_{i_1},\ldots,\beta_{im_i}$ be the dual basis of $\alpha_{i_1},\ldots,\alpha_{im_i}$. For two codewords $\bc\in C, \bu\in C^{\perp}$, write $c_i=\sum_{j=1}^{m_i}c_{ij}\alpha_{ij}$, $u_i=\sum_{t=1}^{m_i}u_{it}\beta_{it}$. Then
$$\Tr(c_iu_i)=\Tr(\sum_{j=1}^{m_i}c_{ij}\alpha_{ij}\sum_{t=1}^{m_i}u_{it}\beta_{it})=
\sum_{j=1}^{m_i}c_{ij}u_{ij}$$

By definition, we have $\sum_{i=1}^n\Tr(c_iu_i)=0$. Thus $\sum_{i=1}^n\Tr(c_iu_i)=\sum_{i=1}^n\sum_{j=1}^{m_i}c_{ij}u_{ij}=0$.
This implies that $\dim_{\F_q} C^\bot=\sum_{i=1}^nm_i-\dim_{\F_q}C$.
\end{proof}

To construct an MOA from codes over mixed alphabets, we need the following definition.
\begin{defn}
Let $C$ be an $\mathbb{F}_{q}$-linear  code with $C\subseteq\bigoplus_{i=1}^n\mathbb{F}_{q^{m_i}}$, we say $G$ is a code matrix if rows of $G$ consist of all codewords of $C$. We say that $H$ is a dual matrix of $C$ if rows of $H$ consist of all codewords of $C^\bot$.
\end{defn}

Now we are going to show that the code matrix $G$ of $C$ is exactly an MOA. We have the following result on G.
\begin{lemma} \label{lem:3.2}
Let C be an $\F_{q}$-linear code of $\bigoplus_{i=1}^n\mathbb{F}_{q^{m_i}}$. Let $G$ be a code matrix of C. If the distance of the code $C^{\perp}$ is $d^\bot\geq k+1$, then for every $A\subseteq[n]$ with $|A|=k$, every vector $\bv_{{A}}$ in $\bigoplus_{i\in A}\mathbb{F}_{q^{m_i}}$ appears exactly $|C|/\prod_{i\in A}q^{m_i}$ times in $G_A$, where $G_A$ is the projection of $G$ at $A$.
\end{lemma}

\begin{proof}
For each vector $\bv_{{A}}$ in $\bigoplus_{i\in A}\mathbb{F}_{q^{m_i}}$, to show that the number of each vector $\bv_{{A}}$ appears exactly $|C|/\prod_{i\in A}q^{m_i}$ times , it is equivalent to prove that the set $$T\triangleq\{\bx_{\bar{A}}\in\bigoplus_{i\in \bar{A}}\mathbb{F}_{q^{m_i}}:(\bv_A,\bx_{\bar{A}})\in C\}$$
has size $|C|/\prod_{i\in A}q^{m_i}$.
Let $\{\bh^{(1)},\bh^{(2)},\ldots,\bh^{(r)}\}$ be an $\mathbb{F}_{q}$-linear basis of $C^\bot$ where $r=\dim_{\F_q} C^{\perp}$.
Write $\bv_A=(v_i)_{i\in A},\bx_{\bar{A}}=(x_i)_{i\in\bar{A}},\bh^{(j)}=((h_i^{(j)})_{i\in A},(h_i^{(j)})_{i\in\bar{A}})$. Then

\begin{eqnarray}
T&=&\{\bx_{\bar{A}}\in\bigoplus_{i\in\bar{A}}\mathbb{F}_{q^{m_i}}:\sum_{i\in A}Tr_i(v_ih_i^{(j)})+\sum_{i\in\bar{A}}Tr_i(x_ih_i^{(j)})=0,j=1,2,\cdots,r\}\\
&=&\{\bx_{\bar{A}}\in\bigoplus_{i\in\bar{A}}\mathbb{F}_{q^{m_i}}:\sum_{i\in\bar{A}}Tr_i(x_ih_i^{(j)})=-\sum_{i\in A}Tr_i(v_ih_i^{(j)}),j=1,2,\cdots,r\}.
\end{eqnarray}
Write $x_i=\sum_{l=1}^{m_i}x_{il}\alpha_{il},h_i^{(j)}=\sum_{l=1}^{m_i}h_{il}^{(j)}\beta_{il}$ where $x_{il}$ and $h_{il}^{(j)}$ are elements in $\F_q$ for $1\le l\le m_i$. Then
$$
|T|=
\begin{vmatrix}
\begin{Bmatrix}
(x_{i1},\ldots,x_{im_i})_{(i\in\bar{A})}\in\mathbb{F}_q^{m_i}:
\begin{pmatrix}
h_{i1}^{(1)}&\ldots&h_{im_i}^{(1)}\\
\vdots&\ddots&\vdots\\
h_{i1}^{(r)}&\ldots&h_{im_i}^{(r)}
\end{pmatrix}
_{i\in \overline{A}}
\begin{pmatrix}
x_{i1}\\
\vdots\\
x_{im_i}
\end{pmatrix}
_{i\in\bar{A}}
=
\begin{pmatrix}
-\sum_{i\in A}Tr_i(v_ih_i^{(1)})\\
\vdots\\
-\sum_{i\in A}Tr_i(v_ih_i^{(r)})
\end{pmatrix}
\end{Bmatrix}
\end{vmatrix}
$$
\end{proof}
where the matrix
\[
\begin{pmatrix}
h_{i1}^{(1)}&\ldots&h_{im_i}^{(1)}\\
\vdots&\ddots&\vdots\\
h_{i1}^{(r)}&\ldots&h_{im_i}^{(r)}
\end{pmatrix}_{i\in \overline{A}}\]
represents \[
\begin{pmatrix}
h_{i_11}^{(1)}&\ldots&h_{i_1m_{i_1}}^{(1)}&\ldots& h_{i_{n-k}1}^{(1)}&\ldots&h_{i_{n-k}m_{i_{n-k}}}^{(1)}\\
\vdots&\ddots&\vdots & \ddots &\vdots&\ddots&\vdots\\
h_{i_11}^{(r)}&\ldots&h_{i_1m_{i_1}}^{(r)}&\ldots& h_{i_{n-k}1}^{(1)}&\ldots&h_{i_{n-k}m_{i_{n-k}}}^{(1)}
\end{pmatrix}\]
for $\overline{A}=\{i_1,\cdots,i_{n-k}\}$.
Now we are going to show that the above matrix has rank $r$.
 Suppose the matrix has rank smaller than $r$, then there exists $\lambda_j$ such that $\sum_{j=1}^r \lambda_j(h_i^{(j)})_{i\in \bar{A}}=\bo$ where $\lambda_j$ are not all zeros for $j=1, \cdots, r$. Recall that $\{\bh^{(1)},\bh^{(2)},\ldots,\bh^{(r)}\}$ is an $\mathbb{F}_{q}$-linear basis of $C^\bot$. Then $\sum_{j=1}^r \lambda_j\bh^{(j)}$ is a codeword of $C^\bot$. Write  $\sum_{j=1}^r \lambda_j\bh^{(j)}=(\sum_{j=1}^r \lambda_j(h_i^{(j)})_{i\in A},\sum_{j=1}^r\lambda_j (h_i^{(j)})_{i\in \bar{A}})$.
As $d^\bot\geq k+1$, by deleting the corresponding positions indexed by $A$ on $\sum_{j=1}^r \lambda_j\bh^{(j)}$, the remaining vector $(\sum_{j=1}^r\lambda_j (h_i^{(j)})_{i\in \bar{A}})$ could not be a zero vector since $|A|=k$. This gives a contradiction. Thus the result that the matrix  has rank $r$ follows.
 This implies the number of solutions is $q^{\sum_{i\in\bar{A}}m_i-r}=q^{\sum\limits_{i=1}^n m_i-\sum_{i\in A}m_i-r}=\frac{|C|}{q^{\sum_{i\in A} m_i}}$ i.e., $|T|=\frac{|C|}{\prod_{i\in A}q^{m_i}}$.

One can easily check that the code matrix $G$ actually leads to an irredundant MOA.
\begin{lemma} \label{lem:3.3}
Let C be an $\mathbb{F}_q$-linear code of $\bigoplus_{i=1}^n\mathbb{F}_{q^{m_i}}$. Let G be a code matrix of C. If $d\geq k+1$, then for any $A\subseteq[n]$ with $|A|=k$, every two rows of the submatrix $G_{\cA}$ are district.
\end{lemma}
\begin{proof}
By assumption, the distance $d$ of $C$ is $\ge k+1$, this means any two codewords are different in more than $k+1$ positions. Thus by deleting $k$ positions, the remaining vector of length $n-k$ are different.
\end{proof}

By Lemmas \ref{lem:3.2} and \ref{lem:3.3}, we can easily see that the code matrix $G$ is exactly an IrMOA $(r,(q^{m_1},\ldots,q^{m_n}),k)$, where $r=q^{\dim_{\F_q}(C)}$ which is the number of rows of $G$.
Thus we can construct a $k$-uniform state from our codes in $\bigotimes_{i=1}^n\CC^{d_i}$ where $d_i=q^{m_i}$.
\begin{thm} \label{thm:3.1}
Let C be an $\mathbb{F}_{q}$-linear code of $\bigoplus_{i=1}^n\mathbb{F}_{q^{m_i}}$. If $\min\{d,d^\bot\}\geq k+1$, then there exists a k-uniform state in $\bigotimes_{i=1}^n\CC^{d_i}$ where $d_i=q^{m_i}$.
\end{thm}
\begin{proof}
The desired result follows from Theorem \ref{thm:2.1}.
\end{proof}

As mentioned, the construction of k-uniform states with high uniformity in heterogeneous systems is an interesting but challenging problem. By Theorem 7, to construct $k$-uniform states with $k$ as large as possible ($k\le\lfloor n/2\rfloor$), we wish to construct codes with distance $d$ and the dual distance $d^{\perp}$ as large as possible.

\subsection{$k$-uniform states from generalized algebraic geometry codes}
In this section, we show that as long as there is a code $C$ over mixed alphabets with distance $d$ and dual distance $d^{\perp}$, we can construct a $k$-uniform state with $k=min{d,d^{\perp}}-1$.
Here we consider a special class of codes called algebraic geometry codes. Firstly we generalize the definition of algebraic geometry codes to the case where coordinates of the codewords may belong to different finite fields. Then by using such generalized algebraic geometry codes, we are able to construct some $k$-uniform states in heterogeneous systems.

Let  $P_1,\cdots,P_n$ be places of degree $m_i$ for $i=1,2,\cdots,n$.
Define the generalized functional algebraic geometry code as follows
$$
C_\cL(\PP,G)=\{(f(P_1),\ldots,f(P_n)):f\in \cL(G)\}, where\;\PP=\{P_1,\ldots,P_n\}.
$$
Define another generalized algebraic geometry code
$$
C_\Omega(\PP,G)=\left\{(\res_{P_1}(\eta),\ldots,\res_{P_n}(\eta)):\eta\in \Omega\left(G-\sum_{i=1}^nP_i\right)\right\}
$$
where $\res_{P_i}(\eta)$ stands for the residue of $\eta$ at $P_i$. We can see that the codes $C_\cL(\PP,G)$ and $C_\Omega(\PP,G)$ are $\F_q$-vector space in $\bigoplus_{i=1}^n\mathbb{F}_{q^{m_i}}$ where $m_i=\deg P_i$. So $C_\cL(\PP,G)$ and $C_\Omega(\PP,G)$ are $\F_q$-linear codes of length $n$ in $\bigoplus_{i=1}^n\mathbb{F}_{q^{m_i}}$.
We have the following lemma on $C_\cL(\PP,G)$ and $C_\Omega(\PP,G)$.

\begin{lemma}\label{lem:4.1}
The dual code of $C_\cL(\PP,G)$ is $C_\Omega(\PP,G)$.
\end{lemma}
\begin{proof}
By the residue Theorem, we have
$$
\sum_{i=1}^nTr_i(f(P_i)\res_{P_i}(\eta))=\sum_{i=1}^nTr_i(\res_{P_i}(f\eta))=0
$$
Furthermore,
$$
\dim_{\mathbb{F}_q}\cL(G)+\dim_{\mathbb{F}_q}\Omega(G-\sum_{i=1}^np_i)=\deg(\sum_{i=1}^nP_i)=\sum_{i=1}^nm_i
$$
The desired result follows.\end{proof}

Below we give a lower bound on the minimum distance of $C_\cL(\PP,G)$ and $C_\Omega(\PP,G)$ respectively. We give two notations first. For the set $\{m_1,m_2,\dots,m_n\}$ and an integer $r\geq1$, denote by
\[t_r(m_1,m_2,\dots,m_n)=\max\{|A|: A\subseteq [n], \sum_{i\in A}m_i\leq r\}.\]
Denote by
\[s_r(m_1,m_2,\dots,m_n)=\min\{|A|: A\subseteq [n], \sum_{i\in A}m_i\ge r\}.\]


\begin{lemma}\label{lem:4.2}
The minimum distance of $C_\cL(\PP,G)$ is at least $n-t_m(m_1,m_2,\dots,m_n)$.
\end{lemma}
\begin{proof}
Let $\bc=(f(p_1),\ldots,f(p_n))$ be a nonzero codeword of $C_\cL(\PP,G)$ with $f\in \cL(G)$. Then $f\neq0$.
Let $A=\supp(\bc)$. Then for $i\in \bar{A}$, $f(P_i)=0$ which means that
$f\in \cL(G-\sum_{i\in\bar{A}}P_i)$. Thus $\deg(G-\sum_{i\in\bar{A}}P_i)\geq0$, i.e.,
$m\geq\sum_{i\in\bar{A}}m_i$. Then we have
$|\bar{A}|\leq t_m(m_1,m_2,\dots,m_n)$ and $|{A}|\geq n-t_m(m_1,\dots,m_n)$.
Therefore, the minimum distance $d(C_\cL(\PP,G))$ of $C_\cL(\PP,G)$ is at least $n-t_m(m_1,m_2,\dots,m_n)$.
\end{proof}

\begin{lemma}\label{lem:4.3}
The minimum distance of $C_\Omega(\PP,G)$ is at least $S_{m-2\g+2}(m_1,\dots,m_n)$.
\end{lemma}
\begin{proof}
Let $\bc=(res_{P_1}(\eta),\ldots,res_{P_n}(\eta))\in C_\Omega(\PP,G)$ with a nonzero $\eta\in\Omega(G-\sum_{i=1}^np_i)$. Let $A=\supp(\bc)$. Then $\eta\in\Omega(G-\sum_{i\in A}P_i)$. So we have $2\g-2\geq deg(G-\sum_{i\in A}P_i)=m-\sum_{i\in A}m_i$. This means
$\sum_{i\in A}m_i\geq m-2\g+2$. Then by definition, $|A|\geq S_{m-2\g+2}(m_1,\dots,m_n)$. The desired result follows.
\end{proof}

Now we are going to present some explicit examples. By choosing suitable curves and points, we can compute the value of $t_r(m_1,\dots,m_n)$ and $s_r(m_1,\dots,m_n)$ explicitly. Combining with Theorem \ref{thm:3.1}, we can show that we are able to produce some $k$-uniform states from generalized algebraic geometry codes. To be simple, we consider the rational function field first, i.e., $\g=0$.
\begin{exm}
Choose $\PP=\{P_1,P_2,\cdots,P_n\}$ where $P_i$ are points of degree one. Then  we know that
 both of $C_\cL(\PP,G)$ and $C_\Omega(\PP,G)$ are Reed-Solomon codes of length $n$ over $\F_q$.
 Suppose $C_\cL(\PP,G)$ is a  Reed-Solomon code of dimension $k$ where $k\leq\frac{n}{2}$. Then the distance $d=n-k+1\geq k+1$ and the distance of the dual code $C_\Omega(\PP,G)$ is $k+1$. Therefore, by Theorem \ref{thm:3.1} we get a $k$-uniform state in the homogeneous system $\bigotimes_{i=1}^n\mathbb{C}^q$ for any $k\leq\frac{n}{2}$.
\end{exm}

\begin{Remark}
Example 1 is reduced to the homogenous case since we only use points of degree one, which is exactly the Reed-Solomon code. This result has been obtained in \cite[Theorem 7]{Feng}. One can also obtain this result from  \cite[Cosntruction 1]{Pang1}.
\end{Remark}

\begin{exm}
Let $n=q+1+\frac{q^2-q}{2}$.
Let $P_1,\ldots,P_{q+1}$ be places of degree one in $\mathbb{F}_q(x)$ and $P_{q+2},\ldots,P_n$ be places of degree two in $\mathbb{F}_q(x)$. Therefore, $\{m_1,\ldots,m_n\}=\{\underbrace{1,\ldots,1}_{q+1},\underbrace{2,\ldots,2}_{{(q^2-q)}/{2}}\}$.
For $q+1\leq r\leq n$, we have $t_r(m_1,\ldots,m_n)=q+1+\lfloor {(r-q-1)}/{2} \rfloor$, $S_r(m_1,\ldots,m_n)=\lceil \frac{r}{2} \rceil$.
Consider the code $C_\cL(\PP,G)$ with $\deg(G)=m$. Then $d\geq n-t_m(m_1,\ldots,m_n)$, $d^\bot\geq S_{m+2}(m_1,\ldots,m_n)$.
By Theorem \ref{thm:3.1}, we want $\min\{d,d^{\perp}\}=\min\{n-t_m(m_1,\ldots,m_n),S_{m+2}(m_1,\ldots,m_n)\}\ge k+1$ for some $k$, namely
$$
 i.e., \left\{
           \begin{array}{ll}
             n-(q+1+\lfloor \frac{m-(q+1)}{2} \rfloor)\geq k+1\\
             \lceil \frac{m+2}{2} \rceil\geq k+1.
            \end{array}.
          \right.
$$

If q is odd, we can take $m=n-\frac{q+1}{2}-1$ and $k=\lceil \frac{n-\frac{q+1}{2}-1+2}{2} \rceil-1=\lceil \frac{n-\frac{q+1}{2}-1}{2} \rceil=\lceil \frac{2n-q-3}{4} \rceil$.
If q is even, we can take $m=n-\frac{q}{2}-2$ and $k=\lceil \frac{n-\frac{q}{2}}{2} \rceil-1=\lceil \frac{2n-q-4}{4} \rceil$.

In conclusion, we have a k-uniform state in $\otimes_{i=1}^{q+1}\mathbb{C}^q\otimes_{i=1}^{\frac{q^2-q}{2}}\mathbb{C}^{q^2}$.where\\
$$
k=
\begin{cases}
\lceil \frac{2n-q-3}{4} \rceil,\quad \text if\; q \;is\; odd\\
\lceil \frac{2n-q-4}{4} \rceil, \quad \text if \;q \;is \;even\\
\end{cases}
$$

Below we give some numerical examples.
\begin{enumerate}
  \item [(1)] Take $q=4$, then $n=11$, we can obtain a $4$-uniform state in $(\CC^4)^{\otimes 5}\otimes(\CC^{16})^{\otimes 6}$.

  \item [(2)] Take $q=5$, then $n=16$, we can obtain a $6$-uniform state in $(\CC^5)^{\otimes 6}\otimes(\CC^{25})^{\otimes 10}$.
\end{enumerate}
\end{exm}

\begin{Remark}
The $k$-uniform state in Example 2 can be produced from \cite[Theorem 3.7]{Pang}.
\end{Remark}

It is interesting to produce $k$-uniform states in $n$-partite quantum states with $k$ as large as possible for given $n$, i.e., where $k=\lfloor\frac{n}{2}\rfloor$ is the optimal case which is called an AME state. We can see that our construction is a powerful method for producing $k$-uniform states since
 the multipartitle quantum states constructed are highly entangled, namely $k$ is quite close to $\frac{n}{2}$.
In fact, we can also produce some AME states from our construction. Note that in literature, there are only few results on the existence of AME states.
The following example shows that there exist a lot of new AME states in heterogeneous systems for odd integer $n$.
\begin{exm}
Let $n$ be an odd integer. Let $P_1,\ldots,P_{n-1}$ be places of degree one in $\mathbb{F}_q(x)$. Let $P_{n}$ be a place of degree two in $\mathbb{F}_q(x)$. Then set $\{m_1,\ldots,m_{n-1},m_n\}=\{1,\ldots,1,2\}$.
For $1\le r\le n-1$, we have $t_r(m_1,\ldots,m_n)=r$, $S_r(m_1,\ldots,m_n)=r-1$.
We want $n-t_m(m_1,\ldots,m_n)\geq k+1$ and $S_{m+2
}(m_1,\ldots,m_n)\geq k+1$.
Take $m=\frac{n-1}{2}$, we can get an $(r=q^\frac{n+1}{2},(q,q,\cdots,q,q^2),\frac{n-1}{2})$ IrMOA, then we  obtain  $k$-uniform state with $k=\lfloor\frac{n}{2}\rfloor=\frac{n-1}{2}$ in $\otimes_{i=1}^{n-1}\mathbb{C}^q\otimes\mathbb{C}^{q^2}$ which is an AME state.

Here we give some numerical examples.
\begin{enumerate}
  \item [(1)] Take $q=3,n=5$, we get a $(27,(3,3,3,3,9),2)$ IrMOA, which can produce a $2$-AME state in $(\mathbb{C}^3)^{\otimes 4}\otimes\mathbb{C}^{9}$.
  \item [(2)] Take $q=5,n=7$, we can produce a $3$-AME state in $(\mathbb{C}^5)^{\otimes 6}\otimes\mathbb{C}^{25}$.
  \item [(3)] Take $q=7,n=9$, we can produce a $4$-AME state in $(\mathbb{C}^7)^{\otimes 8}\otimes\mathbb{C}^{49}$.
\end{enumerate}
\end{exm}

The above examples are obtained from  codes over rational function field, i.e., $\g=0$. We can also consider more general function fields. Below we give an example which is obtained via elliptic curve, i.e., $\g=1$. Of course, we can give many examples from  curves with high genus.

\begin{exm}
Let $q\ge 4$ be an even power of a prime. Choose an elliptic curve over $\F_q$ with $q-2\sqrt{q}+1$ rational places of degree one. Then this curve has $\frac{q^2+2q+1-(q-2\sqrt{q}+1)}{2}=\frac{q^2+q+2\sqrt{q}}{2}$ places of degree 2. Let $n=q-2\sqrt{q}+1+\frac{q^2+q+2\sqrt{q}}{2}$. Let $P_1,\ldots,P_{q-2\sqrt{q}+1}$ be places of degree one in $\mathbb{F}_q(x)$. Let $P_{q-2\sqrt{q}+2},\ldots,P_{n}$ be places of degree two. Then for all $r$ with $q-2\sqrt{q}+1<r<n$. We have
 $t_r(m_1,\ldots,m_n)=q-2\sqrt{q}+1+\lfloor\frac{r-q+2\sqrt{q}-1}{2}\rfloor$, $S_r(m_1,\ldots,m_n)=\lceil\frac{r}{2}\rceil$.
Consider the code $C_\cL(\PP,G)$ with $\deg(G)=m$. To find the largest positive integer $k$ satisfying
$$
 i.e., \left\{
           \begin{array}{ll}
             n-(q-2\sqrt{q}+1+\lfloor\frac{r-q+2\sqrt{q}-1}{2}\rfloor)\geq k+1\\
             \lceil \frac{m}{2} \rceil\geq k+1
            \end{array},
          \right.
$$
we can take $m=n-\lceil\frac{q-2\sqrt{q}+1}{2}\rceil$, then $k=\lceil\frac{n}{2}-\lceil\frac{q-2\sqrt{q}+1}{4}\rceil\rceil-1$.

We also give some numerical examples.
\begin{enumerate}
  \item [(1)] Take $q=4$ then $n=13$, we can produce a $5$-uniform state in $\mathbb{C}^4\otimes(\mathbb{C}^{16})^{\otimes 12}$.
  \item [(2)] Take $q=9$ then $n=52$, we can produce a $24$-uniform state in $(\mathbb{C}^9)^{\otimes 4}\otimes(\mathbb{C}^{81})^{\otimes 48}$.

\end{enumerate}
\end{exm}

To our knowledge, the $k$-uniform states in Example 4 are new since they can not be obtained from the known constructions.

\section{Construction from matrices}

In Section 3, we present a construction of $k$-uniform states in heterogeneous systems  $\CC^{d_1}\otimes \CC^{d_2}\otimes\cdots\otimes \CC^{d_n}$. The idea is to construct IrMOA from codes over mixed alphabets. However, in this construction the levels $d_i$ should be a prime power. In this section, we introduce a totally new method of producing
$k$-uniform states in heterogeneous systems  from  matrices where the levels $d_i$ are more flexible.
The idea is to generalize the symmetric matrix construction in \cite{Feng} from homogeneous systems to heterogeneous systems.

Recall that we denote $G= \bigoplus_{i\in [n]} \ZZ_{d_i}$ and $G_A=\sum_{i\in A}\bigoplus \ZZ_{d_i}$. We can write $G=G_A\bigoplus G_{\bar{A}}$.
Let $D=$ lcm$(d_1,\ldots,d_n)$ and $D_{ij}=\frac{D}{\gcd(d_i,d_j)}$. Suppose $d_1\geq\ldots\geq d_n\geq2$.
For a vector $\bc=(c_1,c_2,\ldots,c_n)\in G$, denote $\bc=(\bc_A,\bc_{\bar{A}})$, where $\bc_A=(c_i)_{i\in A}$, $\bc_{\bar{A}}=(c_i)_{i\in\bar{A}}$.
Let $\zeta_d$ denote a $d$-th primitive root of unity in $\CC$. For two subsets $A,B\subseteq \{1,2,\dots,n\}$ and a matrix $H=(h_{ij})\in \MM_{n\times n}(\ZZ_D)$, we denote by  $H_{A\times B}$ the submatrix $(h_{ij})_{i\in A,j\in B}$.

An $n$-qudit quantum state
$|\Phi\rangle=\sum_{\bc\in G}\phi_{\bc}|\bc\rangle$ in $\CC^{d_1}\otimes \CC^{d_2}\otimes\cdots\otimes \CC^{d_n}$ is associated with a map $\phi$ from $G$ to $\CC$ given by $\phi(\bc)=\phi_{\bc}$.
 Thus,  an $n$-qudit quantum state $|\Phi\rangle$ can be written as $\sum_{\bc\in G}\phi(\bc)|\bc\rangle$ for a given function $\phi$.
As a preparation, we need to give an equivalent definition on $k$-uniform state which is similar with \cite[Lemma 1]{Feng}.
\begin{thm}\label{lem:5.1}
$|\Phi\rangle$ is a k-uniform state if and only if the following conditions are satisfied $(k\geq 1)$£º
\begin{itemize}
    \item [1.] $\phi$ is not identical to zero.
    \item [2.] For any subset $A\subseteq [n]$, $\bc_A,\bc'_A\in G_A$, we have
$$\sum\limits_{c_{\bar{A}}\in G_{\bar{A}}}\bar{\phi}(\bc_A,\bc_{\bar{A}})\phi(\bc_A',\bc_{\bar{A}})=
\begin{cases}
0,\quad\mbox{if $\bc_A \neq \bc'_A$}\\
\frac{1}{|G_A|},\quad\mbox {if $\bc_A=\bc'_A$}.
\end{cases}
$$
\end{itemize}

\end{thm}

Now we are going to give two constructions, i.e., Construction I and Construction II.
Let us take a brief overview for these two constructions. Consider a map $\phi$ from $G$ to $\CC$ given by $\phi(\bc)=\phi_{\bc}$.
For an $n$-qudit state $|\Phi\rangle=\frac{1}{\sqrt{D}}\sum_{\bc\in G }\phi(\bc)|\bc\rangle$ in $\bigotimes_{i=1}^n\CC^{d_i}$  with $\phi(\bc)=\zeta_D^{\bc H\bc^T}$ where $D=$lcm$(d_1,\ldots,d_n)$, we manage to show that if the submatrix $H_{A\times \bar{A}}+H_{\bar{A}\times A}^T$ of $H$ has full row rank for each $k$-subset $A\subseteq \{1,\ldots,n\}$ and $\bar{A}=\{1,\ldots,n\}\setminus A$, then this $n$-qudit state is $k$-uniform. If $H$ is symmetric, this condition can be simplified.
Construction I utilizes the symmetric matrix $H$ to obtain $1$-uniform states and $k$-uniform states for some particular parameters.
Construction II removes the constraint that $H$ is symmetric. It turns out that if $H$ is a random matrix and the underlying alphabet size is big enough, with high probability the submatrix $H_{A\times \bar{A}}+H_{\bar{A}\times A}^T$ of $H$ has full row rank. This provides $k$-uniform states with wide range of parameters.

\subsection{Construction I: from symmetric matrices}


The following Theorem shows that as long as there is a symmetric matrix satisfying certain condition, we can produce $k$-uniform states.
\begin{thm}\label{thm:symmetric}
If there is a zero diagonal symmetric matrix $H=(D_{ij}h_{ij})\in\MM_{n\times n}(\ZZ_D)$ such that for any subset $A\in [n]$ of size $k$, the map $\bc_A H_{A,\cA}=(\sum_{i\in A}c_iD_{ij}h_{ij})_{j\in \cA}\in \ZZ_D^{n-k}$ with domain $\bc_A\in G_A$ is an injection, then the  $n$-qudit state $|\Phi\rangle=\frac{1}{\sqrt{D}}\sum_{\bc\in G }\phi(\bc)|\bc\rangle$ is $k$-uniform in $\bigotimes_{i=1}^n\CC^{d_i}$  with $\phi(\bc)=\zeta_D^{\bc \widetilde{H}\bc^T}$, where $\widetilde{H}=(D_{ij}\widetilde{h}_{ij})_{n\times n}$ with $\widetilde{h}_{ij}=h_{ij}$ for $i<j$ and $0$ otherwise.
\end{thm}
\begin{proof}
 Consider the map $f:G\rightarrow\ZZ_D$ given by $f(\bc)= {\bc \widetilde{H}\bc^T}$. Then for every subset $A$ of  $\{1,2,\cdots,n\}$ with $|A|=k$, we have
\begin{eqnarray*}
&&f(\bc_A,\bc_{\cA})\\
&=&\bc_A(\widetilde{H}_{A\times {\cA}}+\widetilde{H}_{{\cA}\times A}^{T})\bc_{\cA}^T+\bc_A\widetilde{H}_{A\times A}\bc_A^T+\bc_{\cA}\widetilde{H}_{{\cA}\times {\cA}}\bc_{\cA}^T \\ &=&\bc_A H_{A\times {\cA}}\bc_{\cA}^T+\bc_A\widetilde{H}_{A\times A}\bc_A^T+\bc_{\cA}\widetilde{H}_{{\cA}\times {\cA}}\bc_{\cA}^T.
\end{eqnarray*}
Hence,
\begin{eqnarray*}\label{eq:2.1}&&f(\bc_A,\bc_{\cA})-f(\bc'_A,\bc_{\cA})\\&=&(\bc_A-\bc'_A){H}_{A\times {\cA}}\bc_{\cA}^T+\bc_A\widetilde{H}_{A\times A}\bc_A^T-\bc'_A\widetilde{H}_{A\times A}(\bc'_A)^T.\end{eqnarray*}
If $\bc_A=\bc'_A$, one has
\begin{eqnarray*}\sum_{\bc_{\cA}\in G_{\cA}}\overline{\phi(c_A,c_{\cA})}\phi(c'_A,c_{\cA})&=&\frac{1}{D}\sum_{\bc_{\cA}\in G_{\cA}}\zeta_D^{f(\bc_A,\bc_{\cA})}
\zeta_D^{-f(\bc_A,\bc_{\cA})}\\&=& \frac{1}{D}\prod_{i\in \cA} d_i= \frac{1}{\prod_{i\in A}d_i}=\frac{1}{|G_A|}.
\end{eqnarray*}

If $\bc_A\neq\bc'_A$, then there exists a $j\in \cA$ such that $e:=\sum_{i\in A}(c_i-c'_i)D_{ij}h_{ij}\in \ZZ_D$ is nonzero.
Moreover, since $d_jD_{ij}=D\frac{d_j}{\gcd(d_i,d_j)}$ for every $i$, this implies $e\neq 0\pmod {d_j}$ and $\zeta_D^{ e}=\zeta_{d_j}^{e'}$ with $e'=\frac{ed_j}{D}$.
Denote by $(\bc_A-\bc'_A)H_{A,\cA}$ and $\cA'=\cA-\{j\}$. Then,
\begin{eqnarray*}\sum_{\bc_{\cA}\in G_{\cA}}\overline{\phi(c_A,c_{\cA})}\phi(c'_A,c_{\cA})&=&\sum_{\bc_{\cA}\in G_{\cA}}\zeta_D^{(\bc_A-\bc'_A){H}_{A\times {\cA}}\bc_{\cA}}\\
&=&\frac{1}{D}\sum_{\bc_{\cA'}\in G_{\cA'}}\sum_{c_j\in \ZZ_{d_j}}\zeta_D^{(\bc_A-\bc'_A){H}_{A\times {\cA'}}\bc_{\cA'}+ec_j}\\
&=&\frac{1}{D}\sum_{\bc_{\cA'}\in G_{\cA'}}\zeta_D^{(\bc_A-\bc'_A){H}_{A\times {\cA'}}\bc_{\cA'}}\sum_{c_j\in \ZZ_{d_j}}\zeta_{d_j}^{e'c_j}=0.
\end{eqnarray*}
\end{proof}
\begin{Remark}
When $d_1=\ldots=d_n$, we return to the homogenous case~\cite{Feng}. If $d_1,\ldots,d_n$ are mutually co-prime, we have $D_{ij}=D$ and thus $H$ is a zero matrix. This construction is not applicable in this case. Assume $d_1\geq d_2\geq \cdots \geq d_n$ and let $A=\{1,\ldots,k\}$. The condition $H_{A\times \cA}$ has full row rank implies that $\prod_{i=1}^k d_i\leq \prod_{i=k+1}^n d_i$.
\end{Remark}
Next, we provide two examples for $1$-uniform states based on Theorem \ref{thm:symmetric}.

\begin{thm}\label{thm:4.1.2}
Assume $d_1\geq d_2\geq \cdots \geq d_n\geq 2$. If for all $i\in [n]$, it holds that
$d_i \mid$ lcm$(\frac{d_j}{\gcd(d_j,D_{ij})})_{j\neq i}$. Then, there exists a $1$-uniform state in $\CC^{d_1}\otimes \CC^{d_2}\otimes\cdots\otimes \CC^{d_n}$.
\end{thm}
\begin{proof}
We let $H=(D_{ij}h_{ij})$ such that $h_{ij}=1$ for $i\neq j$ and $0$ otherwise.
Since $k=1$, the condition $H_{A\times \cA}$ has full row rank becomes that the equation $xD_{ij}=0 \bmod {d_{j}}$ for any $j\neq i$ has unique solution $x=0$ in $\ZZ_{d_i}$. The solution to $xD_{ij}=0 \bmod {d_{j}}$ satisfies that $\frac{d_j}{\gcd(d_j,D_{ij})}\mid x$ for any $j\neq i$. As $d_i \mid$ lcm$(\frac{d_j}{\gcd(d_j,D_{ij})})_{j\neq i}$, we conclude that $d_i\mid x$ and thus $x=0$ in $\ZZ_{d_i}$. Since this holds for any $i\in [n]$, the proof is completed.
\end{proof}

\begin{thm}\label{thm:4.1.3}
Assume $d_1\geq d_2\geq \cdots \geq d_n\geq 2$. If $d_1=d_2=D=$lcm$(d_1,d_2,\cdots,d_n)$, there exists a $1$-uniform state in $\CC^{d_1}\otimes \CC^{d_2}\otimes\cdots\otimes \CC^{d_n}$.
\end{thm}
\begin{proof}
We let $H=(D_{ij}h_{ij})$ such that $h_{ij}=1$ for $i\neq j$ and $0$ otherwise. It suffices to prove that the equation $xD_{ij}=0 \bmod {d_{j}}$ for any $j\neq i$ has exactly one solution $x=0$ in $\ZZ_{d_i}$. If $i\geq 3$, we let $j=1$ and the solution to $xD_{i1}=0 \bmod {d_{1}}$ satisfies that $\frac{d_1}{\gcd(d_1,D_{i1})}\mid x$. As $d_1=D$ and $D_{i1}=\frac{D}{\gcd(d_1,d_i)}=\frac{D}{d_i}$, we have $d_i\mid x$ and thus there is unique solution $x=0$ in $\ZZ_{d_i}$. If $i=1$, let $j=2$ and we have $D_{12}=1$. This implies that the solution to the equation $xD_{i1}=0 \bmod {d_{2}}$ is unique in $\ZZ_{d_1}$. The same argument can be applied to the case $i=2$ and $j=1$. The proof is completed.
\end{proof}

\begin{exm}
Theorem \ref{thm:4.1.3} shows that there exists a $1$-uniform states for any number of parties as long as $d_1=d_2=D=lcm\{d_1,d_2,\cdots,d_n\}$ where $d_1\geq d_2\geq \cdots \geq d_n\geq 2$.
 \begin{enumerate}
  \item [(1)] There exist AME states in $\CC^d\otimes \CC^d\otimes  \CC^{d'} (2\le d'|d)$.
  \item [(2)] Take $d_1=d_2=12$, $d_3=3,d_4=4$ and we can produce $1$-uniform state in $\CC^{12}\otimes \CC^{12}\otimes \CC^{3}\otimes \CC^{4}$.
  \item [(3)] Take $d_1=d_2=15$, $d_3=3,d_4=3,d_5=5$ and we can produce $1$-uniform state in $(\CC^{12})^{\otimes 2}\otimes (\CC^{3})^{\otimes 2}\otimes \CC^{5}$
     \end{enumerate}
\end{exm}
\

\begin{Remark}
Though 1-uniform states can be easily obtained from IrMOA since any IrMOAs with strength 1 can be easily obtained, the quantum states given here is different from the one obtained from IrMOA.
\end{Remark}
Our last construction yields a $k$-uniform state for $k\geq 2$.
\begin{thm}
Assume $k\geq 2$, $d_1=d_2=\cdots=d_{2k}=d$ and $2\leq d_{2k+i}\mid d$ for $i=1,\ldots,t$. Then, if $\gcd(d,(n-1)!)=1$, there exists a $k$-uniform state in $\bigotimes_{i=1}^n\CC^{d_i}$ with $n=2k+t$. For $t=0$, we obtain the homogeneous AME state.
For $t=1$, we obtain the heterogeneous AME states.
\end{thm}
\begin{proof}
Define the zero diagonal symmetric matrix $H=(D_{ij}h_{ij})$ such that
$h_{ij}=\sum_{a=0}^{k-1}i^aj^{k-1-a}$ for $1\leq i\neq j\leq n$ and $0$ otherwise.
By Theorem \ref{thm:symmetric}, it suffices to show that $\bc_A H_{A,\cA}$ is an injection for any $A\subseteq [n]$ of size $k$.
Let $A=\{i_1,\ldots,i_k\}$ be a $k$-subset of $[n]$.
Then, we have $|\cA \cap [2k]|\geq k$. Let $B=\{j_1,\ldots,j_k\}\subseteq \cA\cap [2k]$ be any subset of size $k$ and we have $d_{j_1}=\cdots=d_{j_k}=d$. If we manage to show that $\bc_A H_{A,B}$ is an injection, then $\bc_A H_{A,\cA}$ is an injection as $B\subseteq \cA$.
The solution of $\bc_A H_{A,B}=0$ satisfies that
$$
0=\sum_{a=1}^k c_{i_a}D_{i_aj_b}h_{i_aj_b}=\sum_{a=1}^k c_{i_a} \frac{d}{d_{i_a}}h_{i_aj_b} \mod d.
$$
for $b=1,\ldots,k$.
Observe that
$$\tilde{H}:=(h_{ij})_{i\in A,j\in B}=\left(
                  \begin{array}{cccc}
                    i_1^{k-1} & i_1^{k-2} & \cdots & 1 \\
                   i_2^{k-1} & i_2^{k-2} & \cdots & 1 \\
                    \vdots & \vdots & \ddots & 1 \\
                    i_k^{k-1} & i_k^{k-2} & \cdots & 1 \\
                  \end{array}
                \right)\left(
                         \begin{array}{cccc}
                           1 & 1 & \cdots & 1 \\
                           j_1 & j_2 & \cdots & j_{k} \\
                           \vdots & \vdots & \ddots & \vdots \\
                           j_1^{k-1} & j_2^{k-1} & \cdots & j_k^{k-1} \\
                         \end{array}
                       \right).
$$
The determinant of the first matrix is $t_1:=\prod_{1\leq a\neq b\leq n}(i_a-i_b)$ and that of the second one is $t_2:=\prod_{1\leq a\neq b\leq n}(j_a-j_b)$. Since $-(n-1)\leq i_a-i_b\leq n-1$ and $\gcd(d,(n-1)!)=1$, $t_1$ is invertible in $\ZZ_d$. The same conclusion holds for $t_2$. Therefore, $\tilde{H} \pmod{d}$ is invertible. Since $d_{i_a}\mid d$ for $a=1,\ldots,k$, $(c_{i_1}\frac{d}{d_{i_1}},\ldots,c_{i_k}\frac{d}{d_{i_k}})\tilde{H}=0 \pmod d$ has a unique solution for $c_{i_a}\in \ZZ_{d_{i_a}}$, $a=1,\ldots,k$. This completes the proof.
\end{proof}

\begin{Remark}
This can be obtain in \cite[Theorem 3.7]{Pang}
\end{Remark}

\subsection{Construction II: from random matrices}
We present Construction II in two steps. At the first step, we assume that $d_i=p^{\alpha_i}$ for some prime number $p$ and positive integer $\alpha_i$. Then, we consider the general case $d_i=\prod_{j=1}^t p_j^{\alpha_{ij}}$.
Our construction can be seen as a generalization of the construction in \cite{Feng} where they assume $\alpha_i$ are the same. We remove the requirement that $H$ is a symmetric matrix because we want to show the existence of such matrix via a probabilistic argument. However, the spirit is the same as Construction I, i.e., the submatrices should have full row rank.
There is a $\ZZ_p$-linear isomorphism $\psi$ from $\ZZ_{p^a}$ to $\ZZ_{p}^{a}$. $\psi$ can be naturally extended to a bijection: $\ZZ_{p^a}^m \rightarrow \ZZ_{p}^{am}$. Thus, we obtain a $\ZZ_p$-linear isomorphism $\psi$ from $G=\bigoplus_{i\in [n]}\ZZ_{p^{\alpha_i}}$ to $\ZZ_{p}^{\alpha }$ with $\alpha=\sum_{i=1}^n \alpha_i$. Since we map each element in $G$ to a vector of length $\alpha$ over $\ZZ_p$, we redefine the index of matrix $H=(h_{ij})\in\MM_{\alpha \times \alpha }(\ZZ_p)$ for convenience. The row of $H$ is indexed by $[i,j]$ with $i\in [n]$ and $j\in [\alpha_i]$ and so does the column index.
The submatrix $H_{A\times\cA}$ is obtained by taking all the entries whose row index $[i_1,j_1]$ satisfies $i_1\in A$ and $j_1\in [\alpha_{i_1}]$ and column index $[i_2,j_2]$ satisfies $i_2\in \cA$ and $j_2\in [\alpha_{i_2}]$.
\begin{thm}\label{thm:primecase}
Assume $d_i=p^{\alpha_i}$ for some prime number $p$ and positive integer $\alpha_i$.
If there is a matrix $H=(h_{ij})\in\MM_{\alpha \times \alpha }(\ZZ_p)$ such that for any subset $A\in \binom{[n]}{k}$, the submatrix $H_{A\times\cA}+H_{\cA\times A}^T$ has full row rank, then the $n$-qudit state $|\Phi\rangle=\frac{1}{p^{\alpha/2}}\sum_{\bc\in G }\phi(\bc)|\bc\rangle$ is $k$-uniform in $\bigotimes_{i=1}^n\CC^{d_i}$ with $\phi(\bc)=\zeta_p^{\psi(\bc)H\psi(\bc)^T}$.
\end{thm}
\begin{proof}
 Consider the map $f$ from $G$ to $\ZZ_p$ given by $f(\bc)= {\psi(\bc) H\psi(\bc)^T}$. Then for every subset $A$ of  $\{1,2,\cdots,n\}$ with $|A|=k$, we have
\begin{eqnarray*}
&&f(\bc_A,\bc_{\cA})\\
&=&\psi(\bc_A)(H_{A\times {\cA}}+H_{{\cA}\times A}^{T})\psi(\bc_{\cA})^T+\psi(\bc_A)H_{A\times A}\psi(\bc_A)^T+\psi(\bc_{\cA})H_{{\cA}\times {\cA}}\psi(\bc_{\cA})^T
\end{eqnarray*}
\begin{eqnarray*}
&&f(\bc_A,\bc_{\cA})-f(\bc'_A,\bc_{\cA})\\&=&(\psi(\bc_A)-\psi(\bc'_A))(H_{A\times {\cA}}+H_{{\cA}\times A}^{T})\psi(\bc_{\cA})^T+\psi(\bc_A)H_{A\times A}\psi(\bc_A)^T-\psi(\bc'_A)H_{A\times A}\psi(\bc'_A)^T.
\end{eqnarray*}
If $\bc_A=\bc'_A$ and thus $\psi(\bc_A)=\psi(\bc'_A)$, one has
\begin{eqnarray*}\sum_{\bc_{\cA}\in G_{\cA}}\overline{\phi(c_A,c_{\cA})}\phi(c'_A,c_{\cA})&=\frac{1}{p^\alpha}&\sum_{\bc_{\cA}\in G_{\cA}}\zeta_p^{f(\bc_A,\bc_{\cA})}
\zeta_p^{-f(\bc_A,\bc_{\cA})}\\&=& \frac{p^{\sum_{i\in \cA}\alpha_i}}{p^\alpha}= \frac{1}{p^{\sum_{i\in A}\alpha_i}}=\frac{1}{|G_A|}.
\end{eqnarray*}
Otherwise, $(\psi(\bc_A)-\psi(\bc'_A))(H_{A\times {\cA}}+H_{{\cA}\times A}^{T})$ is a non-zero vector over $\ZZ_p$, as $H_{A\times {\cA}}+H_{{\cA}\times A}^{T}$ has full row rank.
Using the fact that $\psi$ is a bijection, $(\psi(\bc_A)-\psi(\bc'_A))(H_{A\times {\cA}}+H_{{\cA}\times A}^{T})\psi(\bx)=b$ has exactly $p^{(\sum_{i\in \cA}\alpha_i-1)}$ solutions in $\bigoplus_{i\in \cA} \ZZ_{p^{\alpha_i}}$ for every $b\in \ZZ_p$. Thus, we have
\[
\sum_{\bc_{\cA}\in G_{\cA}}\overline{\phi(\bc_A,\bc_{\cA})}\phi(\bc'_A,\bc_{\cA})=\zeta_p^g\sum_{\bc_{\cA}\in G_{\cA}}\zeta_p^{(\psi(\bc_A)-\psi(\bc'_A))(H_{A\times {\cA}}+H_{{\cA}\times A}^{T})\psi(\bc_{\cA})^{T}}=p^{(\sum_{i\in \cA}\alpha_i-1)}\zeta_p^g\sum_{b=0}^{p-1}\zeta_p^{b}=0,\]
where $g=\psi(\bc_A)H_{A\times A}\psi(\bc_A)^T-\psi(\bc'_A)H_{A\times A}\psi(\bc'_A)^T$. This completes the proof.
\end{proof}

To show the existence of such matrix in Theorem \ref{thm:primecase}, we resort to the probabilistic method.
\begin{thm}\label{thm:randomprime}
Assume $d_i=p^{\alpha_i}$ with prime $p$. Then there exists an $n$-qudit $k$-uniform quantum state in $\bigotimes_{i=1}^n\CC^{d_i}$ if
\begin{equation*}
\sum_{A\in \binom{[n]}{k}}p^{(\sum_{i\in A}\alpha_i-\sum_{i\in \cA} \alpha_i-1)}<1.
\end{equation*}
\end{thm}
\begin{proof}
Let $H=(h_{ij})\in\MM_{\alpha \times \alpha }(\ZZ_p)$ be a uniformly random matrix. Let $X_A$ be a binary random variable such that
$X_A=1$ if $H_{A\times {\cA}}+H_{{\cA}\times A}^{T}$ has full row rank and zero otherwise.
We note that the entries of $H_{A\times {\cA}}$ and $H_{{\cA}\times A}$ are distinct random variables. Thus, $H_{A\times {\cA}}+H_{{\cA}\times A}^{T}$ is a $\sum_{i\in A}\alpha_i\times \sum_{i\in \cA}\alpha_i$ uniformly random matrix over $\ZZ_p$.
The probability that the $\sum_{i\in A}\alpha_i$ rows of $H_{A\times {\cA}}+H_{{\cA}\times A}^{T}$ are $\ZZ_p$-linearly dependent is $p^{(\sum_{i\in A}\alpha_i-\sum_{i\in \cA} \alpha_i-1)}$. By the linearity of expectation, we have
$$E[\sum_{A\in \binom{[n]}{k}}X_A]=\sum_{A\in \binom{[n]}{k}}p^{(\sum_{i\in A}\alpha_i-\sum_{i\in \cA} \alpha_i-1)}.$$
If such value is less than $1$, there exists a matrix $H$ such that for all $A\in \binom{[n]}{k}$, the submatrix $H_{A\times \cA}$ has full row rank. The proof is completed.
\end{proof}
\begin{cor}\label{cor:primecase}
Assume $d_i=p^{\alpha_i}$ with prime $p$ and $\alpha_1\geq \alpha_2\geq\cdots\geq \alpha_n$.
If $\sum_{i=1}^k \alpha_i\leq \sum_{i=k+1}^{n} \alpha_i$ and $p>\binom{n}{k}$, there exists
a $k$-uniform state in $\bigotimes_{i=1}^n\CC^{d_i}$. If $n=2k+1$, this implies the existence of AME states.
\end{cor}
\begin{proof}
The condition $\alpha_1\geq \alpha_2\geq\cdots\geq \alpha_n$ implies that
$\sum_{i\in A} \alpha_i-\sum_{i\in \cA}\alpha_i\leq \sum_{i=1}^k \alpha_i-\sum_{i=k+1}^{n} \alpha_i$ for any $A\subseteq \binom{[n]}{k}$. Therefore, the probability constraint in Theorem \ref{thm:randomprime} can be simplified as
$$
\sum_{A\in \binom{[n]}{k}}p^{(\sum_{i\in A}\alpha_i-\sum_{i\in \cA} \alpha_i-1)}\leq \binom{n}{k}p^{(\sum_{i=1}^k \alpha_i- \sum_{i=k+1}^{n} \alpha_i-1)}\leq \binom{n}{k}\frac{1}{p}<1.
$$
The proof is completed.
\end{proof}


We proceed to the general case that $d_i=\prod_{j=1}^t p_j^{\alpha_{ij}}$ with prime number $p_1,\ldots,p_t$ and non-negative integer
$\alpha_{ij}$. Let $\alpha_i=\max\{\alpha_{i1},\ldots,\alpha_{it}\}$ and $\alpha=\sum_{i=1}^n \alpha_i$. Let $d=\prod_{j=1}^t p_i$.  Since $d_i=\prod_{j=1}^t p_i^{\alpha_{ij}}$ can divide
$d^{\alpha_i}$, we can embed $\ZZ_{d_i}$ into $\ZZ_{d^{\alpha_i}}$. There is a $\ZZ_{d}$-linear isomorphism $\psi$ from $\ZZ_{d^{m}}$ to $\ZZ_{d}^{m}$. Similarly, we abuse the notation by extending $\psi$ to an isomorphism from $\ZZ_{d^{m}}^k$
 to $\ZZ_d^{km}$.
We index the row of a matrix $H$ with $[i,j]$ such that $i\in [n]$ and $j\in [\alpha_i]$ and so does the column index. The submatrix $H_{A_a\times\cA_a}$ is obtained from $H$ by taking all the entries whose row index $[i_1,j_1]$ satisfies $i_1\in A$ and $j_1\in [\alpha_{i_1a}]$ and column index $[i_2,j_2]$ satisfies $i_2\in \cA$ and $j_2\in [\alpha_{i_2a}]$.
\begin{exm}
Assume $d_1=2^2\times 3=12$ and $d_2=2\times 3^2=18$. This implies that $p_1=2,p_2=3$ and $\alpha_{11}=2,\alpha_{12}=1,\alpha_{21}=1,\alpha_{22}=2$.
Given an element $(7,10)\in \ZZ_{12}\oplus \ZZ_{18}$, we obtain $\psi(7)=(1,1)\in \ZZ_6^2$ and $\psi(10)=(4,1)\in \ZZ_6^2$.
In our following argument, we want to decompose $\psi: \ZZ_{12}\rightarrow \ZZ_{6}^2$ into two maps $\psi_1:\ZZ_{12}\rightarrow \ZZ_{2}^2$ and
$\psi_2:\ZZ_{12}\rightarrow \ZZ_{3}^2$ and so does $\psi: \ZZ_{18}\rightarrow \ZZ_{6}^2$.
This implies $\psi_1(7)=(1,1)\in \ZZ_2^2$, $\psi_{2}(7)=(1,0)\in \ZZ_3^2$ and $\psi_1(10)=(0,0)\in \ZZ_2^2$, $\psi_{2}(10)=(1,0)\in \ZZ_3^2$.
We emphasize that although $\psi_2(7)$ and $\psi_2(10)$ are the same, the domain of these two maps are different.
The domain of the former one is $\ZZ_{12}$ and the latter one is $\ZZ_{18}$. Our intuition is to decompose $\psi$ so that Theorem \ref{thm:primecase} works. We abuse the notation of the map when such information is clear. This also means the second coordinate of $\psi_2(\ZZ_{12})$ is always zero while that of $\psi_2(\ZZ_{18})$ could be any value.
This will not affect our argument in the following theorem as we will remove those zero coordinates. It can be seen from the submatrix $H_{A_a\times\cA_a}$ when we pick in the statement of the following theorem.

\end{exm}
\begin{thm}\label{thm:generalcase}
Assume $d_i=\prod_{j=1}^t p_j^{\alpha_{ij}}$ with prime number $p_1,\ldots,p_t$ and non-negative integer
$\alpha_{ij}$. Let $D=\prod_{i=1}^n d_i$.
If there is a matrix $H=(h_{ij})\in\MM_{\alpha \times \alpha }(\ZZ_d)$ such that for any subset $A\in \binom{[n]}{k}$ and $a\in [t]$,  $(H_{A_a\times\cA_a}+H_{\cA_a\times A_a}^T)\pmod {p_a}$ has full row rank over $\ZZ_{p_a}$, then the  $n$-qudit state $|\Phi\rangle=\frac{1}{\sqrt{D}}\sum_{\bc\in G }\phi(\bc)|\bc\rangle$ is $k$-uniform in $\bigotimes_{i=1}^n\CC^{d_i}$ with $\phi(\bc)=\zeta_d^{\psi(\bc)H\psi(\bc)^T}$.
\end{thm}
\begin{proof}
Define the map $f$ from $G$ to $\ZZ_p$ given by $f(\bc)= {\psi(\bc) H\psi(\bc)^T}$. Then for every subset $A$ of  $\{1,2,\cdots,n\}$ with $|A|=k$,
the similar argument can show that
\begin{eqnarray*}
&&f(\bc_A,\bc_{\cA})-f(\bc'_A,\bc_{\cA})\\&=&(\psi(\bc_A)-\psi(\bc'_A))(H_{A_a\times\cA_a}+H_{\cA_a\times A_a}^T)\psi(\bc_{\cA})^T+\psi(\bc_A)H_{A\times A}\psi(\bc_A)^T-\psi(\bc'_A)H_{A\times A}\psi(\bc'_A)^T.
\end{eqnarray*}
If $\bc_A=\bc'_A$ and thus $\psi(\bc_A)=\psi(\bc'_A)$, one has
\begin{eqnarray*}\sum_{\bc_{\cA}\in G_{\cA}}\overline{\phi(c_A,c_{\cA})}\phi(c'_A,c_{\cA})&=&\frac{1}{D}\sum_{\bc_{\cA}\in G_{\cA}}\zeta_d^{f(\bc_A,\bc_{\cA})}
\zeta_d^{-f(\bc_A,\bc_{\cA})}\\&=& \frac{\prod_{i\in \cA}d_i}{D}= \frac{1}{\prod_{i\in A}d_i}=\frac{1}{|G_A|}.
\end{eqnarray*}

Otherwise, $(\psi(\bc_A)-\psi(\bc'_A))$ is a non-zero vector over $\ZZ_d$ where $d=\prod_{i=1}^t p_i$. If $(\psi(\bc_A)-\psi(\bc'_A)) \pmod {p_a}$ is nonzero, then by condition we have $(\psi(\bc_A)-\psi(\bc'_A))(H_{A_a\times\cA_a}+H_{\cA_a\times A_a}^T) \pmod {p_a}$ is a nonzero vector\footnote{In fact, we can treat this map as $\psi_a$ as we define in the above example}. This implies that $(\psi(\bc_A)-\psi(\bc'_A))(H_{A_a\times\cA_a}+H_{\cA_a\times A_a}^T)\psi(\bx)=b$ has exactly $p_a^{(\sum_{i\in \cA}\alpha_{ia}-1)}$ solutions in $\bigoplus_{i\in \cA} \ZZ_{p_a^{\alpha_{ia}}}$ for every $b\in \ZZ_{p_a}$.

Let $S=\{i\in [n]: p_i | \psi(\bc_A)-\psi(\bc'_A)\}$ and $p_S=\prod_{i\in S}p_i$, $p_{\bar{S}}=\frac{d}{p_S}$ . From above argument and the Chinese Remainder Theorem, we know that for every $b\in \ZZ_{p_{\bar{S}}}$, the number of solutions $(\psi(\bc_A)-\psi(\bc'_A))(H_{A_a\times\cA_a}+H_{\cA_a\times A_a}^T)\psi(\bx)=b$ is exactly $\prod_{a\in S}p_a^{(\sum_{i\in \cA}\alpha_{ia}-1)}$ for $ \bx\in \bigoplus_{i\in \cA}\ZZ_{\prod_{a\in S} p_a^{\alpha_{ia}}}$. In other words, the number of solutions $(\psi(\bc_A)-\psi(\bc'_A))(H_{A_a\times\cA_a}+H_{\cA_a\times A_a}^T)\psi(\bx)=b \pmod {p_{\bar{S}}}$ is exactly $\frac{\prod_{i\in \cA} d_i}{p_S}$ for $\bx\in G_{\bar{A}}=\bigoplus_{i\in \cA}\ZZ_{d_i}$.
Thus, we have
\[
\sum_{\bc_{\cA}\in G_{\cA}}\overline{\phi(\bc_A,\bc_{\cA})}\phi(\bc'_A,\bc_{\cA})=\zeta_d^g\sum_{\bc_{\cA}\in G_{\cA}}\zeta_{p_{\bar{S}}}^{\big(\frac{1}{p_S}\big)(\psi(\bc_A)-\psi(\bc'_A))(H_{A_a\times\cA_a}+H_{\cA_a\times A_a}^T)\psi(\bc_{\cA})^{T}}=\frac{\prod_{i\in \cA} d_i}{p_S}\zeta_d^g\sum_{b=0}^{p_{\bar{S}}-1}\zeta_{p_{\bar{S}}}^{b}=0,\]
where $g=\psi(\bc_A)H_{A\times A}\psi(\bc_A)^T-\psi(\bc'_A)H_{A\times A}\psi(\bc'_A)^T$. This completes the proof.
\end{proof}
\begin{thm}
Assume $d_i=\prod_{j=1}^t p_j^{\alpha_{ij}}$. Then there exists an $n$-qudit $k$-uniform quantum state in $\bigotimes_{i=1}^n\CC^{d_i}$ if
\begin{equation*}
\sum_{j=1}^t\sum_{A\in \binom{[n]}{k}}p_j^{(\sum_{i\in A}\alpha_{ij}-\sum_{i\in \cA} \alpha_{ij}-1)}<1.
\end{equation*}
\end{thm}
\begin{proof}
Let $H=(h_{ij})\in\MM_{\alpha \times \alpha }(\ZZ_d)$ be a uniformly random matrix. From Theorem \ref{thm:generalcase}, it suffices to show that for any subset $A\in \binom{[n]}{k}$ and $a\in [t]$,  $(H_{A_a\times\cA_a}+H_{\cA_a\times A_a}^T)\pmod {p_a}$ has full row rank.
We fix $a\in [t]$ and reuse the argument in Theorem \ref{thm:randomprime}. This gives us the claim that for any subset $A\in \binom{[n]}{k}$, $(H_{A_a\times\cA_a}+H_{\cA_a\times A_a}^T)\pmod {p_a}$ has full row rank with probability $1-\sum_{A\in \binom{[n]}{k}}p_a^{(\sum_{i\in A}\alpha_{ia}-\sum_{i\in \cA} \alpha_{ia}-1)}$. Taking a union bound over all $a\in [t]$, we obtain the desired result.
\end{proof}

\begin{cor}
Assume that $d_i=\prod_{j=1}^t p_i^{\alpha_{ij}}$ and $\alpha_{j_1j}\geq \alpha_{j_2j}\geq\cdots\geq \alpha_{j_nj}$ for $j\in [n]$.
If $\sum_{i=1}^k \alpha_{j_ij}\leq \sum_{i=k+1}^{n}\alpha_{j_ij}$ and $p_j>t\binom{n}{k}$, then
there exists a $k$-uniform state in $\bigotimes_{i=1}^n\CC^{d_i}$. If $n=2k+1$, this implies the existence of AME-uniform states.
\end{cor}
\begin{proof}
This is the generalization of Corollary \ref{cor:primecase}. The argument is exactly the same except that we apply the argument to each $i\in [t]$.
\end{proof}

From above corollary, we can obtain a large family of $k$-uniform states. Here we give some examples.
\begin{exm}
We have the following $k$-uniform states.
\begin{enumerate}
\item [(1)] There exists a $4$-uniform state in $(\CC^{36})^{\otimes 10}\otimes (\CC^{18})^{\otimes 3}\otimes \CC^{8}\otimes \CC^{27}\otimes \CC^6$.
\item [(2)] There exists a $4$-uniform state in $(\CC^{225})^{\otimes 9}\otimes \CC^{125}\otimes \CC^{9}\otimes \CC^{45}\otimes \CC^{15}$.
\item [(3)] There exists a $5$-uniform state in $(\CC^{15})^{\otimes 17}\otimes (\CC^{3})^{\otimes 3}\otimes \CC^{5}$.
\item [(4)] There exists a $5$-uniform state in $(\CC^{59\times 61})^{\otimes 11}\otimes \CC^{59^2}\otimes \CC^{61^2}$.
\end{enumerate}
\end{exm}

\section{Conclusions}
In this paper, we investigate the existence and construction of $k$-uniform states in heterogeneous systems. We present two general constructions of $k$-uniform states. The first construction is derived from error correcting codes. We manage to show that if the distance and dual distance of an error correcting code is large than $k+1$, we can convert such an error correcting code to a $k$-uniform state. Contrary to the classical codes, our error correcting codes allow mixed alphabets which implies the existence of $k$-uniform states of heterogeneous systems.
In \cite{Feng}, a matrix construction is given for producing $k$-uniform states of homogenous systems. In Section 4, we  generalize this matrix construction to the heterogeneous systems.
This construction guarantees the existence of $k$-uniform states if a matrix $H$ meet the condition that the submatrix $H_{A\times \bar{A}}+H_{\bar{A}\times A}^T$ of $H$ has full row rank for each $k$-subset $A\subseteq \{1,\ldots,n\}$ and $\bar{A}=\{1,\ldots,n\}\setminus A$. If $H$ is symmetric, this condition can be simplified.
Then, we apply a probabilistic method to show the existence of such matrix which produces a large family of $k$-unform states.

\end{document}